\documentclass[11pt]{article} 

\usepackage{amsmath,amsthm,latexsym,amssymb,amsfonts,epsfig}

\newtheorem{Theorem}{Theorem}[section]
\newtheorem{Definition}{Definition}[section]

\newcommand{\be}{\begin{equation}}
\newcommand{\ee}{\end{equation}}
\newcommand{\ba}{\begin{eqnarray}}
\newcommand{\ea}{\end{eqnarray}}

\title{{\sf Exact quantisation of U(1)$^3$ quantum gravity}\\
{\sf via exponentiation of the hypersurface deformation algebroid}} 
\author{
{\sf T. Thiemann}$^1$\thanks{{\sf 
thomas.thiemann@gravity.fau.de}}\\
\\
{\sf $^1$ Inst. for Quantum Gravity, FAU Erlangen -- N\"urnberg,}\\
{\sf Staudtstr. 7, 91058 Erlangen, Germany}\\
}
\date{{\small\sf \today}}

\makeatletter
\@addtoreset{equation}{section}
\makeatother

\begin{document} 

\maketitle

{\sf

\begin{abstract}
The U(1)$^3$ model for 3+1 Euclidian signature general relativity 
is an interacting, generally covariant field theory with two physical
polarisations that shares many features of Lorentzian general relativity. 
In particular, it displays a non-trivial realisation of the hypersurface 
deformation algebroid with non-trivial, i.e. phase space dependent 
structure functions rather than structure constants.

In this paper we show that the model admits {\it an exact quantisation}.
The quantisation rests on the observation that for this model and in the 
chosen representation of the canonical commutation relations 
the density unity 
hypersurface algebra {\it can be exponentiated on non-degenerate 
states}. These are states that represent a non-degenerate quantum 
metric and from a classical perspective are the relevant states 
on which the hypersurface algebra is representable. 

The representation of the algebra is exact, with no ambiguities
involved and anomaly free. The quantum constraints can be exactly solved 
using {\it groupoid averaging} 
and the solutions admit a Hilbert space structure that agrees with the 
quantisation of a recently found reduced phase space formulation. 
Using the also recently found covariant action for that model, we 
start a path integral or spin foam formulation which, due to the 
Abelian character of the gauge group, is much simpler than for 
Lorentzian signature general relativity and provides an ideal testing 
ground for general spin foam models.  

The solution of U(1)$^3$ quantum gravity communicated in this paper 
motivates an entirely new approach to the implementation of the 
Hamiltonian constraint in quantum gravity. 
\end{abstract}

\section{Introduction}
\label{s1}

The initial value formulation of General Relativity (GR) \cite{1} is the 
starting point for both numerical GR \cite{2} producing black hole merger 
templates of ever increasing accuracy \cite{3} and canonical quantisation 
\cite{4}. A central ingredient in this approach to the dynamics of both 
classical and quantum gravity are the initial value constraints, known 
as the spatial diffeomorphism and Hamiltonian constraints. With respect 
to an ADM foliation \cite{5} of spacetime they play a dual role first as 
temporal-spatial and temporal-temporal projections of the Einstein equations 
and second as generator gauge transformations as well as 
dynamical equations for the two physical polarisations 
of the spacetime metric. 

These constraints form a closed algebroid \cite{6} under Poisson brackets known
as the {\it hypersurface deformation algebroid} \cite{7}. The term algebroid 
rather than algebra emphasises the fact that in contrast to a Poisson 
Lie algebra the constraints do not close with structure constants but 
with structure functions that have a non-trivial dependence on the phase 
space point that one is considering. Thus, multiple Poisson brackets 
produce always a linear combination of constraints, however, the coefficients 
of these linear combinations are functions on the phase space which 
become more and more complictated the more Poisson brackets one computes.
This is in contrast to 1-parameter groups of spacetime diffeormorphisms 
whose generating vector fields do form a true algebra and not just an 
algebroid. Indeed it is known that the 1-1 correspondence between 
the canonical constraint algebroid and the spacetime diffeomorphism algebra 
holds only ``on shell'' i.e. when the Einstein equations are satified.

The fact that the constraint algebra of GR is an algebroid and not an algebra 
is one of the many reasons why it continues to be so difficult to construct a 
theory of quantum gravity. In recent decades some progress has been made 
using the connection formulation \cite{8} and gave rise to a quantisation 
programme coined Loop Quantum Gravity (LQG) \cite{9}. The name 
arose because of the similarity of LQG to lattice gauge theory \cite{10} 
in which 
gauge covariant Wilson loop functionals play a fundamental role. While 
it has been possible to define the quantum constraint on a common, dense, 
invariant domain 
\cite{11} of a rigorously defined Hilbert space representation of the 
canonical commutation and $\ast$ relations \cite{12} and while the 
corresponding quantum algebroid indeed closes and in that sense is 
{\it mathematically consistent}, it closes with the 
wrong quantum structure functions \cite{13}, thereby exhibiting 
a {\it physical anomaly}. To improve this, there have been
at least four suggestions: In the master constraint \cite{14}
approach structure 
functions are omitted altogether by a classical equivalent reformulation
of the constraints in terms of a single constraint. In the reduced phase space
approach \cite{15} one solves the constraints classically and thus avoids 
the issue of quantum structure functions. In \cite{20} one uses classically
equivalent constraints which are linear combinations of the original ones 
with phase space dependent coefficients (``electric shifts'' and non-standard 
density weight for the Hamiltonian constraint)  
and tries to define those 
on a space of distributions over a dense subspace of the Hilbert space       
rather than that dense subspace itself. Finally in \cite{17} one uses 
renormalisation techniques to define a renormalisation group flow of structure 
operators whose fixed point should be the physically correct ones.

These four programmes have been tested in various model situations, see e.g. 
\cite{18,19,20,21} and references therein. Here we wish to focus on two 
models, parametrised 
field theory (PFT) in two spacetime dimensions \cite{22} and the 
U(1)$^3$ model for 3+1 Euclidian GR \cite{23}. Both models are much simpler 
than GR but still share with GR the fact that they exhibit a non-trivial 
hypersurface deformation algebra. The constraints of PFT are yet 
much simpler than those of the U(1)$^3$ theory in that in their  
density weight two version they form a true (centrally extended)
algebra rather than an 
algebroid while this is no longer the case for the U(1)$^3$ model 
which is therefore a much better testing ground.
  
The application of \cite{11} to PFT
can be found in \cite{24}, the application of \cite{20} to PFT in \cite{25}
and the application of \cite{17} to PFT in \cite{26}.
The application of \cite{15} to U(1)$^3$
can be found in \cite{27}, the application of \cite{20} to 
U(1)$^3$ in \cite{28}. The work \cite{26} shows that in PFT 
it is possible to find the correct anomaly free fixed point
algebra using renormalisation and the work \cite{30} that one 
can modify the constraints of PFT to have non-trivial structure functions 
with density weight unity constraints which still form a closed quantum 
algebroid on a space of distributions {\it which are non-degenerate} 
in the sense that the quantum metric is diagonal with non degenerate 
eigenvalues. The work \cite{27} shows that the non standard 
density weight is motivated by a space of distributions representing 
degenerate quantum metrics. Note that the work \cite{28} is purely 
classical.

In this paper we show that the U(1)$^3$ model for 3+1 Euclidian GR,
or U(1)$^3$ quantum gravity for short, can 
be solved {\it exactly} following analogous steps as in LQG. What makes this 
possible is the fact that, while the constraints of the 
U(1)$^3$ model display exactly the same algebraic stucture as those for 
3+1 Euclidian gravity, the latter is for the non-Abelian 
gauge group SU(2) rather than the Abelian U(1)$^3$. This has the consequence 
that all constraints are {\it at most linear in the connection rather than 
quadratic}. While this does not turn U(1)$^3$ into a free or topological 
theory, it is in fact highly interacting and displays two propagating 
polarisations, the linearity in the connection makes it possible 
to {\it exponentiate the hypersurface deformation algebroid}. 

In LQG 
one already exponentiated the spatial diffeomorphism constraints. This 
is relatively straightforward as 
those constraints form a closed and true sub Lie algebra of the 
hypersurface deformation algebra. In LQG the spatial diffeomorphisms 
act by unitary operators but 1-parameter subgroups do not act weakly 
continuously. Similarly in this paper we show that the 
exponentiated Hamiltonian constraints act by unitary operators, but 
1-parameter ``subgroupoids'' do not act weakly continuously. Just as 
in the classical theory, while the composition of exponentiated 
quantum Hamiltonian constraint actions can be computed in a form as closed 
as in the classical theory commutators thereof cannot be written 
as an action of an exponentiated spatial diffeomorphism due to the 
structure operators involved. Yet, the resulting expression is precisely
the action of the corresponding classically exponentiated constraints and 
thus is represented without anomaly, directly on Hilbert space, without 
considering any dual spaces and without any ambiguities.
It is on dual spaces that one can compute 
infinitesimal actions and these 
can also been shown to be anomaly free, giving a representation of the 
hypersurface deformation algebroid, without any ambiguities.

The Hilbert space representation considered in this paper for U(1)$^3$ 
theory is similar to but different from the LQG representation. 
It is based on generalised holonomies of the connection. The generalisation
consists in modifying all three ingredients of an LQG spin network function:
Graphs, spins and intertwiners are replaced by divergence free smearing 
functions. This can considered as ``thickening'' the graph edges, to allow 
real valued rather than half integral spin quantum numbers on the edges 
and to 
take care of the Abelian nature of the gauge group by replacing 
invariant intertwiners at vertices by the divergence free condition.

This more general state space is precisely what allows the exponentiation 
of the Hamiltonian constraint. Just as for the spatial diffeomorphism 
constraint, that action makes use of the classical exponentiation, i.e.
the Hamiltonian flow of the corresponding Hamiltonian vector field 
and that flow preserves a suitable space of {\it non-degenerate}
smearing functions.
That flow can be worked out as usual by Taylor expansion and as expected 
acts {\it highly non-linearly} on the space of smearing functions. This 
is in sharp contrast to the exponentiated spatial diffeomorphism constraint
which has a linear action. In fact, for the U(1)$^3$ model the exponentiated 
Hamiltonian constraint flow can be defined {\it for any density weight} 
of the Hamiltonian constraint including density weight one. Then 
the action is not polynomial in the smearing function, not even when 
truncating the Taylor series to finite order. However, no matter 
how non-polynomial that action is, it maps the space of divergence free, 
density weight unity vector fields onto themselves. To the best of knowledge 
of the author, such a non-polynomial representation of the hypersurface 
deformation algebroid has not been discussed previously in the literature 
and it is interesting to see how the Hamiltonian constraint in fact naturally
generates it without any guess work about loop attachments ever necessary.\\
\\
The lessons to be learnt from the present paper for actual GR, in the 
opinion of the author, are as follows:\\
1.\\
The present exposition once again stresses the importance of the quantum 
non-degeneracy condition for a faithful representation of the 
hypersurface deformation algebroid \cite{30}.\\
2.\\
The model shows that there is no obstacle in using the natural density unity 
form of the Hamiltonian constraint. As shown in \cite{30} density 
unity is enforced as soon as one considers as a cosmological constant 
or additional matter terms in GR.\\
3.\\
The model shows that one can obtain anomaly free closure of the constraint
algebra directly on the kinematical Hilbert space without invoking 
dual spaces, in particular the closure is off-shell.\\ 
\\
The architecture of this paper is as follows:\\
\\

In section \ref{s2} we briefly outline the classical description of the 
U(1)$^3$ model both in its covariant and canonical formulation 
\cite{27}. 

Section \ref{s3} is the key section of the present paper. We show that 
any classical constraint linear in momentum $p$ admits a unitary representation
in a Hilbert space representation of the CCR and the $\ast$ relations 
based on a cyclic vacuum $\Omega$ for the Weyl operators depending only 
on $p$ which is annihilated by the conjugate variable $q$. This 
holds no matter how non-linearly the dependence of the constraint on $q$ maybe.
In particular, this allwos to extend U(1)$^3$ QG by a quantum cosmological 
constant. The resulting action of the constraints is free of any 
(ordering) ambiguities.
 
In section \ref{s4} we define such a representation of the CCR and $\ast$ 
relations for the U(1)$^3$ model and apply the theorem to the spatial 
diffeormorphism and Hamiltonian constraint. We compute explicitly the 
first few terms of the Taylor expasion mentioned above and 
demonstrate unitarity, and anomaly freeness of the exponentiated 
constraints.

In section \ref{s5} we compute the dual action of the constraints on 
suitable distributions and verify anomaly freeness of the algebroid.

In section \ref{s6} we solve the quantum constraints by groupoid averaging 
and demonstrate that we arrive precisely at a reduced phase space 
quantisation of the reduced phase space description of \cite{27}. 
In particular we can construct the unitary 1-parameter group generated 
by the physical Hamiltonian. The resulting theory is a kind of 
self-interacting, non-polynomial quantum electrodynamics with two 
propagating polarisations.  

In section \ref{s8} we construct new non-relational, weak Dirac
observables which 
are not related to any gauge fixing condition, both classically and
quantum mechanically. We establish that they 
weakly commute with all quantum constraints.

In section \ref{s7} we develop a path integral formulation of the 
reduced  description of the model that by construction 
is equivalent to the canonical operator theory. An interesting aspect 
is that instead of the usual heuristic undefined Lebesgue measure
expressions, the systematic derivation yields automatically Bohr and 
discrete measures instead.  

In section \ref{s9} we unfold the reduced phase space path integral 
and arrive at a covariant formulation of the rigging map (``projector''
into the constraint kernel)
which depends on the 
exponent of the classical action. 
This is interesting because 
the rigging map is difficult to construct when the constraints 
do not form a Lie algebra \cite{31}. 
We transform that path integral 
over U(1)$^3$ connections and tetrads into a path integrals over 
connections and a $B$ field ($BF$ formulation) which is subject to 
U(1)$^3$ simplicity constraints. This reformulation is 
the starting point 
for a systematic spin foam treatment which should be 
much simpler in this Abelian setting. The fact that this 
model receives a substantial 
amount of guidance from the canonical treatment layed out in this paper
may help to deepen the bridge between canonical and covariant LQG. 
The details of the spin foam formulation of U(1)$^3$ can be worked 
out using the results of the present paper and are reserved for 
future research.

In section \ref{s10} we summarise and conclude. In particular we compare 
the action of the Hamiltonian constraint in the usual LQG representation
with that of the present paper. This is possible because one can 
mollify the form factors of the ususal formulation to arrive at 
form factors of the present formulation. Among other things, 
a main difference between the two actions is that the lapse function
in the present setting becomes {\it part of the form factor in a 
non-linear fashion}, i.e. it is not simply a coefficient in an expansion
of spin (or charge) network functions.

\section{Classical U(1)$^3$ theory}
\label{s2}

Our exposition will be minimal. The details can be found in 
\cite{23,27,28}.\\
\\
A possible action has recently been found in \cite{27}
\be \label{2.1}
\int_{M}\; d^X\; \det(\{e_\rho^K\})\; F_{\mu\nu}^{IJ}\; e^\mu_I\; e^\nu_J
\ee
with spacetime tensor indices $\mu,\nu,\rho=0,1,2,3$ and frame indices 
$I,J,K=0,1,2,3$. The field $e_\mu^I$ is a co-tetrad with inverse 
$e^\mu_I$. The field 
$F_{\mu\nu}^{IJ}$ is not the curvature of the Palatini formulation of 
Lorentzian vacuum GR but rather is constrained by 
\be \label{2.2}
F_{\mu\nu}^{0j}=\frac{1}{2}\delta^{jm}\;
\epsilon_{mkl}\; F_{\mu\nu}^{kl}=:F_{\mu\nu}^j=2\partial_{[\mu} A^j_{\nu]} 
\ee
where $A^j_\mu$ is a spacetime U(1)$^3$ connection with 
$j,k,k=1,2,3$.. If we would 
add to $F_{\mu\nu}^j$ the quadratic term 
$\epsilon^{jkl}\;\delta_{km}\;\delta_{ln}\; A_\mu^m\; A_nu^n$ we would 
obtain precisely the selfdual formulation of Euclidian signature GR
\cite{8} and one would pass from the Abelian group U(1)$^3$ to the 
non-Abelian group SU(2). 
The action (\ref{2.1}) could be generalised by an Immirzi
parameter \cite{32,27} but for the purpose of this paper (\ref{2.1}) will 
be sufficient.

A careful Dirac constraint analysis of (\ref{2.1}) 
reveals \cite{27} that the 3+1 split $M\cong \mathbb{R}\times \sigma$ of the 
action, after getting rid of second class constraints, exhibits the 
following canonical data:\\
There are conjugate pairs $(A_a^j,E^a_j)$ with $a,b,c=1,2,3$ and constraints
\ba\label{2.3}
G_j &=& \partial_a\; E^a_j 
\nonumber\\
D_a &=& F_{ab}\; E^b_j-\; A_a^j G_j   
\nonumber\\
C &=& F_{ab}^j \; \epsilon_{jkl}\; E^a_m\; E^b_n\; \delta^{km}\; \delta^{ln}
\; |\det(E)|^{-1/2}
\ea
which are produced in the Dirac analysis directly in the natural density 
unity form. The constraints generate U(1)$^3$ gauge transformations, 
spatial diffeomorphisms and Hamiltonian transformations respectively. Their 
smeared versions 
\be \label{2.4}
G[s]:=\int_\sigma\; d^3x\; s^j\; G_j,\;\; 
D[u]:=\int_\sigma\; d^3x\; u^a\; D_a,\;\;
C[M]:=\int_\sigma\; d^3x\; M\; C,\;\; 
\ee
yields the Poisson algebra
\ba \label{2.5}
&& \{G[s],G[t]\}=0,\;\;\{D[u],G[s]\}=-G[u[s]],\;\{C[N],G[s]\}=0 
\nonumber\\
&& \{D[u],D[v]\}=-D[[u,v]],\; \{D[u],C[M]\}=-C[u[M]]
\nonumber\\
&& \{C[M],C[N]\}=-D[q^{-1}[M\; dN-N\; dM]
\ea
which is computed using the fundamental Poisson brackets
\be \label{2.6}
\{E^a_j(x),A_b^k(y)\}=\delta^a_b\; 
\delta_j^k\; \delta(x,y)
\ee
and where $s^j,M$ are scalars, $u$ is a vector field, $u[s], u[M]$ is 
the action of a vector field on scalars and $[u,v]$ the commutator vector 
field. Finally and most importantly for the present paper the 
metric field 
\be \label{2.7}
q_{ab}:=e_a^j e_a^k\delta_{jk},\;E^a_j:=\sqrt{\det(q)}\; q^{ab} e_b^j
\ee
is a spatial metric derived from the field $e_a^j$ which appears 
in the action (\ref{2.1}), i.e. $e_\mu^I$ for $\mu=a,\; I=j$. 

The four last 
relations in (\ref{2.5}) show that $D[u],C[M]$ represent the 
{\it hypersurface deformation algebroid (HDA)} $\mathfrak{h}$ \cite{7}.
The algebra does not close with the structure constants (i.e. smearing 
functions are independent of $A,E$) but with structure functions 
encoded in the inverse $(q^{-1})^{ab}=:q^{ab},\; q^{ac}\; q_{cb}=\delta^a_b$
of the metric. The whole chain of relations (\ref{2.7}) and the appearance 
of $q^{-1}$ make it transparent that the very formulation of $\mathfrak{h}$ 
{\it assumes that the metric non-degenerate}, i.e. nowhere singular. 
This is of course a necessary condition of classical GR as otherwise 
signature, curvature etc. would not be well defined on the whole spacetime 
manifold. The classical hypersurface deformation algebra therefore reminds us
of this basic but important {\it non-degeneracy} condition. Also, when 
going through the details of the classical calculation that leads to 
(\ref{2.5}), one makes use of that non-degeneracy in every single 
step of the calculation. One may argue that for the density weight two 
version of the Hamiltonian constraint 
the classical non-degeneracy condition is not necessary. However, as shown 
in \cite{30} only density weight unity is possible in Lorentzian vacuum quantum 
GR or even Euclidan quantum GR with a cosmological constant. As shown 
in \cite{11}, density weight unity is also consistent with any known 
matter coupling.    

Now in LQG \cite{9} this non-degeneracy condition is dealt with 
as follows: The state space of LQG is the closed  
span of (so called spin network) 
functions that exhibit excitations of the quantum metric only 
on the edges of finite graphs, therefore a typical LQG state is actually 
degenerate 
almost everywhere. Nevertheless, 
operators corresponding to $q^{ab}$ that appear in 
(\ref{2.5}) must be densely defined on such states. This is indeed possible 
\cite{11} exploiting singular properties of the LQG volume operator 
\cite{33}. Obviously that operator must then vanish where the metric is 
not excited, i.e. the inverse quantum metric vanishes at zero quantum metric 
while the inverse classical metric diverges at zero classical metric. 
While this comes out somewhat naturally from the LQG framework, it has the 
following far reaching consequence: The last Poisson bracket relation 
in (\ref{2.5}) relies on the fact that the inverse classical metric is 
nowhere vanishing. Would it be non-vanishing almost everywhere then the 
right hand side of that last relation would actually vanish becausse
the Riemann/Lebesgue integral does not ``see'' this set. This 
is precisely what one observes when one computes the dual action of 
the operators $C[M]$ on suitable spaces of distributions \cite{34}.
While this behaviour obviously depends on the selected space of distributions 
and while the commutator on the LQG Hilbert space itself does not vanish
\cite{11} it is clear that a proper representation of $\mathfrak{h}$ in 
the quantum theory must be subject to some kind of {\it quantum
non-degeneracy} in order to circumvent these difficulties if one 
does not want to change the density weight of the Hamiltonian constraint 
\cite{20} which is forbidden by cosmological constant, 
Lorentzian rather than Euclidian signature vacuum terms and matter \cite{30}.
In \cite{30} we show that there exists a non-trivial non-degenerate
sector in the standard LQG Hilbert space, but those states obviously are
no longer finite linear combinations of spin network functions but 
rather live in their closure defining a new but very 
complicated domain for the constraint 
operators. For the U(1)$^3$ model a much simpler solution exists which 
will be presented in the next section.

\section{Quantum U(1)$^3$ theory} 
\label{s3}

We will use some of the properties of LQG but modify them in some important 
details. \\
\\
We first define a Hilbert space representation of the CCR and $^\ast$ 
relations (we work in units with $\hbar=1$ and drop an analog of 
Newton's constant)
\be \label{3.1}
[E^a_j(x),E^b_k(y)]=   
[A_a^j(x),A_b^k(y)]=0,\;
[E^a_j(x),A_b^k(y)]=\delta^a_b\;\delta_j^k\;\delta(x,y),\;
[E^a_j(x)]^\ast-E^a_j(x)=   
[A_a^j(x)]^\ast-A_a^j(x)=0
\ee
Abusing the notation we will denote the representing operators and 
abstract algebra elements (\ref{3.1}) by the same symbol.

The representation is based on a cyclic vector $\Omega$ which is a 
{\it vacuum} 
for the electric field $E^a_j(x)$, that is 
\be \label{3.2}
E^a_j(x)\;\Omega=0
\ee
We will excite the vacuum by {\it Weyl elements}
\be \label{3.3}
w[F]:=\exp(-i\int_\sigma\; d^3x \; F^a_j(x) \; A_a^j(x))
\ee
As a consequence of (\ref{3.1}), (\ref{3.2}) and (\ref{3.3}) 
(more precisely the corresponding Weyl 
relations induced by them)
the excited states
$w[F]\Omega$ are simulatenous eigenstates of the operator valued 
distributions $E^a_j(x)$
\be \label{3.4}
E^a_j(x)\; w[F]\;\Omega=F^a_j(x)\;w[F]\Omega
\ee
and we necessarily have \cite{30}
\be \label{3.5}
<\Omega,\; w[F]\Omega>=\delta_{F,0}
\ee
where $\delta_{F,0}$ is indeed non-vanishing (namely unity) if and only if 
$F^a_j(x)\equiv 0$ thereby displaying discontinuity of the 
excited states and non-separability of this Hilbert space of 
Narnhofer-Thirring type \cite{34}. Together with the Weyl relations 
$w[F]^\ast w[F']=w[F'-F]$ the Hilbert space is the closure of the span of 
these excited states equipped with the inner product induced from (\ref{3.5}).
We see that in this representation it is quite easy to find everywhere 
non-degenerate states because by the spectral theorem 
\be \label{3.6}
q_{ab}(x)\;w[F]\Omega=q^F_{ab}(x) \; w[F]\Omega,\;\;\;
q_{ab}^F=\frac{1}{2|\det(F)|}
\epsilon_{acd}\epsilon_{def}  F^c_j F^e_k F^d_m F^f_n\; 
\delta^{jm}\;\delta^{kn}
\ee
thus as expected 
the eigenvalue $q^F_{ab}$ is everywhere well defined if $\det(F)$ is 
nowhere vanishing. 

At this point one might wonder why one does not proceed the same way in LQG
and chooses to work with the much more complicated space of spin network 
states. There are actually two reasons. One is gauge invariance, the 
other is the fact that the constraints of LQG depend quadratically 
on $A$ and not only linearly. Concerning gauge invariance, in order 
to solve the non-Abelian Gauss constraint one uses non-Abelian holonomies.
It would be very difficult to solve the Gauss constraint using the Weyl
elements (\ref{3.3}). Concerning the quadratic dependence of the 
constraints on the connection we note that the minimal smearing dimension 
of a field is dictated by the dynamics \cite{30}. For the Hamiltonian 
constraint of U(1)$^3$ theory, if we smear $A$ in 3d as in (\ref{3.3}) 
then the electric field dependence is diagonal and the action 
of $C$ is roughly by multiplying a Weyl element 
by a functional linear in $A$ smeared 
in 3d. It thus has the same form as the exponent of $w[F]$ 
and thus has a chance 
to be writable as (limits) of multiplication operators acting on 
the $w[F]\Omega$. If we did the same smearing for the Hamiltonian 
constraint of Euclidian GR we would again get a diagonal electric
field dependence but now the resulting function is not of the 
form of a linear functional of $A$ smeared in 3d but rather a quadratic 
expression in $A$ smeared in 3d. Thus to find a quadratic functional
of $A$ smeared in as many dimensions as the exponent of 
$w[F]$ we must lower the smearing 
dimension of $w[F]$ to $1\le k\le 2$ 
which means that $E$ now no longer multiplies by a function but by 
a $\delta$ distribution in $3-k$ dimensions. In order that the cosmological
term also be well defined with a density weight $\delta$ Hamiltonian 
constraint the unique choice is $w=k=1$ \cite{30}.
The choice $k=1$ is the natural choice taken in LQG from the perspective 
of gauge covariance as the holonomy is anyway
along a 1-dimensional curve.

However, in U(1)$^3$ we have the luxury to use smearing dimension 3 for 
$w[F]$ and thus also $C$ can be considered with any density weight, including 
unity. The density weight $\delta$ valued Hamiltonian constraint is defined 
by 
\be \label{3.6a}
C_\delta=C\; |\det(E)|^{(\delta-1)/2}=
\frac{\epsilon_{abc}\epsilon^{jkl}\; B^a_j\; E^b_k\; E^c_l}
{|\det(E)|^{(2-\delta)/2}}         
\ee
where $B^a=\epsilon^{abc} F_{bc}/2$ is the magnetic field of $A$.

We start by quantising the Gauss constraint which is in fact diagonal on 
the $w[F]$ and imposes the condition that the smearing functions 
be divergence free 
\be \label{3.7}
\partial_a \; F^a_j=0
\ee
That is all there is to do to solve the Gauss constraint, no complicated 
closure constraints as in LQG have to be imposed.

The state of affairs is not as simple with respect to the spatial 
diffeomorphism and Hamiltonian constraints because in contrast to the 
Gauss constraint, which is independent of $A$, 
they depend linearly on $A$. However, the connection smeared 
against a function 
\be \label{3.8}
A[F]=<F,A>=\int_\sigma\; d^3x\; A_a^j\; F^a_j
\ee
is not a well defined operator in the Hilbert space, just $w[F]=e^{-i A[F]}$
is. We could now proceed as in LQG or Loop Quantum Cosmology 
(LQC) \cite{35} and replace this by say 
the regulated expression $A_\epsilon[F]=
\sin(\epsilon A[F])/\epsilon)=i(w[F_\epsilon]-w[-F_\epsilon])/(2\epsilon),\;\;
F_\epsilon=\epsilon F$ but this 
introduces ambiguities as we could use other approximants which are 
subleading as $\epsilon\to 0$ and anyway we cannot take the limit 
$\epsilon\to 0$.

Given the fact that only the {\it exponential} of $A[F]$ is well defined 
it is a natural question to ask whether at least the {\it exponentials}
of $D[u],\; C_\delta[M]$ are well defined. The idea is to regulate 
the $A$ dependence in those constraints in terms of Weyl elements 
of the form $w[F_\epsilon]$ and then to exponentiate those regulated 
constraints. This has the chance to bring the whole $A$ dependence of 
the constraints, which is not defined in its linear version, into the
exponent where it can itself become a (operation on a) Weyl element which 
would be well defined. Now while the Weyl elements $w[F_\epsilon]$ are not 
continuous as operators with respect to $\epsilon$, the functions 
$F_\epsilon$ are continuous in a suitable topology (e.g. as Schwartz 
functions or simply pointswise in the sense of smooth functions). 
Then taking the limit $\epsilon\to 0$ on the set of those functions
serves as a well motivated {\it definition} of the exponentiated and 
regulator free operator. We will first investigate this     
for general constraints linear in the momenta and then turn to 
the concrete U(1)$^3$ model constraints.      
\begin{Theorem} \label{th3.1} ~\\
Consider a phase space with configuration variables $E^I$ where 
$I$ is from any index set and conjugate momenta 
$A_I$ and let $C_M=A_I\; G^I_M(E)$ be phase space functions 
linear in $A$ and of arbitrary dependence in $E$ 
in precisely 
this ordering and labelled by and depending 
linearly on $M$. Consider 
a Hilbert space representation of the CCR and $^\ast$ relations 
with cyclic vacuum $E^I\Omega=0$, dense set of vectors of the 
form $w[F]\Omega,\; w[F]=e^{-i A[F]}, \; A[F]=A_I F^I$ and inner product 
$<\Omega, w[F]\Omega>=\delta_{F,0}$. Suppose that $G_M^I(E)\Omega=0$. Let 
\be \label{3.11}
C_{M,\epsilon}:=\sum_I\; 
\frac{i}{2\epsilon}\{w[\epsilon \;\delta_I]-w[-\epsilon\;h_I]\}
\;G_M^I(E);\;\;\; (h_I)^J:=\delta_I^J
\ee
Then 
\be \label{3.12}
e^{-i\; C_{M,\epsilon}}\;w[F]\Omega=w[F_{M,\epsilon}(F)]\;\Omega
\ee
where 
\be \label{3.13}
\lim_{\epsilon\to 0} F_{M,\epsilon}(F)=[e^{X_M} K](F,0)
\ee
Here $X_M$ is the classical Hamiltonian vector field of $C_M$ and 
$K^I(E,A):=E^I$ the I-the coordinate function.
\end{Theorem}
\begin{proof}:\\
We have by the spectral theorem and due to $G^I_M\Omega=0$ 
\be \label{3.14}
C_{M\epsilon}\;w[F]\Omega=
G_M^I[F]\; 
\frac{i}{2\epsilon}[w[F+\epsilon h_I]-w[F-\epsilon h_I]]\;\Omega
\ee
We introduce multiplication and discrete derivative operations 
\be \label{3.15}
[\hat{G}_M^I \; w](F):=G_M^I(F)\; w[F],\;\;
(\Delta_{\epsilon,I}\; w)[F]:=
\frac{1}{2\epsilon}[w[F+\epsilon h_I]-w[F-\epsilon h_I]]
\ee
on the space of operator valued functionals $F\mapsto\;f[F]$.  
Then 
\be \label{3.16}
C_{M,\epsilon}\;w[F]\Omega=i
(\sum_I\; \hat{G}_M^I\; \Delta_{\epsilon,I}\; w)[F]\;\Omega
=:i\; (X_{M,\epsilon}\; w)[F]\;\Omega
\ee
We now show by induction that 
\be \label{3.17}
C_{M,\epsilon}^n\;w[F]\Omega=i^n\; (X_{M,\epsilon}^n \; w)[F]  
\ee
To see this we write $(X_{M,\epsilon}^n \; w)[F]$ explicitly
\ba \label{3.18}
&& (X_{M,\epsilon}^n \; w)[F]=\frac{1}{(2\epsilon)^n}\;
\sum_{\sigma_1,..,\sigma_n;I_1,..,I_n}\; 
\\
&& G_M^{I_1}(F)\;G_M^{I_2}(F+\epsilon\sigma_1\;h_{I_1})\;...\;
G_M^{I_n}(F+\epsilon(\sigma_1\;h_{I_1}+..+\sigma_{n-1}\; h_{I_{n-1}}))
\;\; w[F+\epsilon(\sigma_1 h_{I_1}+..+\sigma_n h_{I_n})]
\nonumber
\ea
with $\sigma_k=\pm 1,\; k=1,..,n$, which one also can see by induction. 
Using again that $G_M^I(E)$ is diagonal on $w[F]\Omega$ with eigenvalue 
$G_M^I(F)$ and the induction assumption we get
\ba \label{3.19}
&& 
C_{M,\epsilon}^{n+1}\;w[F]\Omega
=C_{M,\epsilon}\;C_{M,\epsilon}^n\;w[F]\Omega=
i^n\; C_{M,\epsilon}\; (X_{M,\epsilon}^n \; w)[F]  
\nonumber\\
&=& \frac{i^{n+1}}{(2\epsilon)^{n+1}}\;
\sum_{\sigma_1,..,\sigma_{n+1};I_1,..,I_{n+1}}\; 
\nonumber\\
&& G_M^{I_1}(F)\;
G_M^{I_2}(F+\epsilon\sigma_1\; h_{I_1})\;...\;
G_M^{I_n}(F+\epsilon(\sigma_1\;h_{I_1}+..+\sigma_{n-1}\; h_{I_{n-1}}))
\;\;C_{M,\epsilon}\;w[F+\epsilon(\sigma_1 h_{I_1}+..+\sigma_n h_{I_n})]
\nonumber\\
&=& \frac{i^n}{(2\epsilon)^n}\;
\sum_{\sigma_1,..,\sigma_n;I_1,..,I_n}\;
\nonumber\\
&& G_M^{I_1}(F)\;
G_M^{I_2}(F+\epsilon\sigma_1\;h_{I_1})\;...\;
G_M^{I_{n+1}}(F+\epsilon(\sigma_1\;h_{I_1}+..+\sigma_n\; h_{I_n}))
\;\; w[F+\epsilon(\sigma_1 h_{I_1}+..+\sigma_{n+1} h_{I_{n+1}n})]
\nonumber\\
&=& i^{n+1}\; (X_{M,\epsilon}^{n+1}\; w)[F]\;\Omega
\ea
It follows using formal Taylor expansion
\be \label{3.20}
e^{-i\; C_{M,\epsilon}}\; w[F]\Omega
=(e^{X_{M,\epsilon}}\; w)[F]\;\Omega
\ee
The operator $w[F]$ on the Hilbert space of square inegrable functions 
with respect to the Bohr measure acts by multiplication 
\be \label{3.21}
(w[F]\psi)(A)=e^{-i A_I\; F^I}\; \psi(A)=:w_A[F]\; \psi(A)
\ee
Therefore (\ref{3.20}) when evaluated at $A$ may also be written 
\be \label{3.22}
\{e^{-i\; C_{M,\epsilon}}\; w[F]\Omega\}(A)
=(e^{X_{M,\epsilon}}\; w_A)[F]\;\Omega(A)
\ee
as $X_{M,\epsilon}$ does not act on $A$. We introduce the coordinate 
function on the classical phase space 
$K^I(E,A)=E^I$ whence $w_A[F]=\exp(-i A_I K^I(F,0))$. We formally 
extend $\Omega$ to be the constant function of $F$ i.e. 
$\Omega(A,F)=\Omega(A)$ so that 
\be \label{3.23}
\{e^{-i\; C_{M,\epsilon}}\; w[F]\Omega\}(A)
=(e^{X_{M,\epsilon}}\; w_A\;\Omega)(A,F)
=(e^{X_{M,\epsilon}}\; w_A\;e^{-X_{M,\epsilon}}\;\Omega)(A,F)
\ee
where $X_{M\epsilon}\Omega(A,F)=0$ was used, i.e. that 
$\Delta_{\epsilon,I}$ annihilates constant functions. Now
\be \label{3.24}
e^{X_{M,\epsilon}}\; w_A\;e^{-X_{M,\epsilon}}
=\exp(-i\; A_I\; e^{X_{M,\epsilon}}\; K^I(.,0)\;e^{-X_{M,\epsilon}})
=\exp(-i\; A_I\;e^{X_{M,\epsilon}}\; K^I(.,0)\;e^{-X_{M,\epsilon}})
\ee
Finally 
\be \label{3.25}
[e^{X_{M,\epsilon}}\; K^I(.,0)\;e^{-X_{M,\epsilon}})](F)
=[\sum_{n=0}^\infty\;\frac{1}{n!}\; [X_{M,\epsilon},K^I(.,0)]_{(n)}]
=\sum_{n=0}^\infty\;\frac{1}{n!}\; (X_{M,\epsilon}^n\;K^I)(F,0)
\ee
which converges pointwise in phase space to 
\be \label{3.26}
(e^{X_M}\; K^I)(F,0)
\ee
i.e. the Hamiltonian flow of $C_M$.\\
\end{proof}
Theorem \ref{th3.1} motivates the following definition.
\begin{Definition} \label{def3.1}
The exponentiated constraints are densely defined by
\be \label{3.27}
U(M) w[F]\Omega:=e^{-i\;C_M} \; 
w[F]\;\Omega:=w[(e^{X_M} K)(F)]\; \Omega
\ee
\end{Definition}
We note that the linearity of $C_M$ in $A$ is essential: It means that the 
Hamiltonian flow $e^{X_M}$ {\it preserves the polarisation} of the phase 
space, i.e. it maps functions of $E$ again to functions of $E$ only. 
This is no longer true for say a quadratic dependence on $A$ which is 
why what follows only applies to the U(1)$^3$ truncation of Euclidian 
GR. On the other hand, the density $\delta$ cosmological term 
$V_M=\Lambda \int\; d^3x\; |\det (E)|^{\delta/2}$ just depends on $E$ 
and thus trivially preserves the $E$ polarisation. We could therefore 
simply add it to $C_M$ and keep the definition (\ref{3.27}). However,
then the cosmological constant contribution completely drops out 
from $e^{X_M}\; K$, thus this cannot be the correct generalisation 
of theorem \ref{th3.1} when an additional potential term is present.
\begin{Theorem} \label{3.1a} ~\\
Keep all assumptions as in theorem \ref{th3.1} except that the 
constraints are generalised by a potential term 
$C_M=A_I\; G_M^I(E)+V_M(E)$ such that also $V_M(E)\Omega=0$. Then 
\be \label{3.27a}
e^{-i C_M}\; w[F]\;\Omega=e^{-i\alpha_M[F]}\; w[(e^{X_M}\; K)(F)]\;\Omega
\ee
where $X_M$ is the Hamiltonian vector field of $A_I G_M^I$ and the 
phase $\alpha_M:=[\alpha_M^s]_{s=1}$ is given by 
\be \label{3.27b}
\alpha^s_M(F)=\int_0^s\; dt\; [V_M(e^{t X_M}\; K)](F)
\ee
\end{Theorem}
\begin{proof}:\\
We follow exactly the same steps as in the proof of theorem \ref{3.1}.
Then we find 
\ba \label{3.27c}
&& e^{-i C_M} \; w[F]\;\Omega
=(e^{X_M-i V_M} \; w[K])_{K=F}\;\Omega
=(e^{X_M-iV_M} \; w[K] \; e^{-(X_M-iV_M)} )_{K=F}\;
[e^{X_M-iV_M}\Omega](K)_{K=F}
\nonumber\\
&=& w[(e^{X_M}\; K)(F)] \;[e^{X_M-iV_M}\Omega](K)_{K=F}
\ea
The difference to the previous situation is that still $X_M\Omega=0$ but 
$(V_M \Omega)(F)=V_M(F)\Omega\not=0$ because $V_M$ is a multiplication
operator of the space of $F$ dependent funactions. We define the second factor 
by the Trotter product (note that we consider the continuous space 
of smearing functions itself, not the discontinuous space of functions
of connections labelled by them) 
\be \label{3.27d}
[e^{s(X_M-iV_M)}\Omega](K)_{K=F}:=\lim_{N\to \infty}\;
([e^{\frac{s}{N}\; X_M}\; e^{-i\frac{s}{N}\; V_M}]^N\; \Omega)(K)_{K=F}
\ee
Using $e^{-\frac{s}{N}X_M}\Omega=\Omega$ and 
\be \label{3.27e}
e^{\frac{s}{N} X_M}\; e^{-i\frac{s}{N}\;V_M(e^{\frac{ks}{N} X_M}\;K)}\; 
e^{-\frac{s}{N} X_M}
=e^{-i\frac{s}{N}\;V_M(e^{\frac{(k+1)s}{N} X_M}\;K)}\; 
\ee
we can compute (\ref{3.27d}) exactly
\be \label{3.27f}
[e^{s(X_M-iV_M)}\Omega](K)_{K=F}=\lim_{N\to \infty}\;
\; [e^{-i\frac{s}{N}\;\sum_{k=1}^N\; V_M(e^{\frac{ks}{N} X_M}\;K)}](F)\;\Omega
=e^{-i[\int_0^s\;dt\; V_M(e^{t X_M}\; K)](F)}\;\Omega
\ee
\end{proof}
Given the fact that the action of the quantum constraints is dictated by their
classical Hamiltonian flow on ``polarised'' functions (in the sense 
of geometric quantisation \cite{44}) we obtain the following expected 
result. 
\begin{Theorem} \label{th3.2} ~\\
The exponentiated constraints have the following properties:\\
1. unitarity $U(M)^\dagger=U(-M)=U(M)^{-1}$\\
2. weak discontinuity\\
3. anomaly freeness (in the sense defined below)
\end{Theorem}
\begin{proof}:\\
1.\\
By the assumed linearity in $M$ we have $X_M=-X_{-M}$ and by elementary 
properties of Hamiltonian vector fields $e^{-X_M}\; e^{X_M}={\rm id}$
is the identity canonical transformation. Next 
\ba \label{3.28a}
&& U(-M)\; U(M)\; w[F]\Omega
=e^{-i\alpha_M(F)}\; U(-M)\; w[(e^{X_M}\;K)(F)]\Omega
\nonumber\\
&= &e^{-i[\alpha_M(F)+\alpha_{-M}((e^{X_M} K)(F))]}\;
\; w[e^{X_M}\;e^{X_{-M}}\;F]\Omega=w[F]\;\Omega
\ea
as, using again linearity in $M$ i.e. $V_{-M}=-V_M$
\be \label{3.28b}
\alpha_{-M}((e^{X_M} K)(F)
=\int_0^1\; dt\; V_{-M}((e^{X_M}\; e^{-t X_M}\;K)(F)
=-\int_0^1\; dt\; V_M((e^{1-t)X_M}\;K)(F)=-\alpha_M(F)
\ee
Thus $U(M)^{-1}=U(-M)$ and $U(M)$
has an inverse on the dense span of the $w[F]\Omega$. Then
\ba \label{3.28}
&& <w[F]\Omega, U(M)\; w[F']\Omega>
=e^{-i\alpha_M(F')}\delta_{F,[e^{X_M} K](F')}
=e^{-i\alpha_M(F')}\;\delta_{[e^{-X_M} K](F),F'}
\nonumber\\
&=& e^{-i\alpha_{M}((e^{-X_M} K)(F))}\;\delta_{[e^{-X_M} K](F),F'}
=e^{i\alpha_{-M}(F)}\;\delta_{[e^{-X_M} K](F),F'}
=<U(-M)\; w[F]\Omega,w[F']\Omega>
\ea
is unitary on the same domain where (\ref{3.28b}) was used again.
The extension to the full Hilbert as a unitary 
operator to the full Hilbert space then is a consequence of the BLT theorem.\\
2.\\
One parameter unitary subgroups are of the form $s\mapsto U(sM)$ with $M$ 
fixed. Then for instance \\
$<w[F]\Omega,\;\; U(sM)\; w[F]\Omega>=\delta_{s,0}$.\\
3.\\
As by item 2. the self-adjoint 
generators $C_M$ of $s\mapsto U(sM)$ do not exist by Stone's theorem, we 
content ourselves by verifying the classically equivalent finite versions 
of the classical closure condition
\be \label{3.29}
\{C_\alpha,C_\beta\}=\kappa_{\alpha,\beta}\;^\gamma \; C_\gamma,\;
C_M=\sum_\alpha\; M^\alpha\; C_\alpha
\ee
where $\kappa_{\alpha,\beta}\;^\gamma$ are structure functions on the 
phase space, i.e. they may depend non-trivially on $E$. Note that they 
cannot depend on $A$ because 
\be \label{3.30}
\{C_M,C_N\}=2\; (A_I\;\{G^I_{[M},A_J\} G^J_{N]}+\{V_{[M},A_I\}\; G_{N]}^I) 
\ee
contains terms at most linear in $A$ which themselves must combine 
to the constraint operators. 
Recall the identity 
\be \label{3.30}
e^{s\; X_M}\; e^{t\; X_N}\; e^{-s\; X_M}\; e{-t\; X_N}
%(1 + sX_M + s^2/2 X_M^2)
%(1 + tX_N + t^2/2 X_N^2)   
%(1 - sX_M + s^2/2 X_M^2)
%(1 - tX_N + t^2/2 X_N^2)   
=1+st\;[X_M, X_N]+O(s^3,s^2 t, s t^2, t^3)
=1+st\;X_{\{C_M-V_M,C_N-V_N\}}+O(s^3,s^2 t, s t^2, t^3)
\ee
We need the {\it composition law} of the $U(M)$
\be \label{3.31}
U(M)\; U(N)\;w[F]\Omega
=e^{-i \alpha_N(F)}\; U(M)\;w[(e^{X_N} K)(F)]\Omega
=e^{-i[\alpha_N(F)+\alpha_M((e^{X_N} K)(F))]}\;
w[(e^{X_M} K)(e^{X_N} K)(F))]\Omega
\ee
Making use of the automorphism property of the Hamiltonian flow for 
a general function $H$ on the phase space and with the coordinate 
function $P_I(E,A)=A_I$
\be \label{3.32}
(e^{X_M} H)(E,A)=[H(e^{X_M} K, e^{X_M} P)](A,E)
\ee
and applied to $H=e^{X_M} K$
\be \label{3.32}
U(M)\; U(N)\;w[F]\Omega
=e^{-i[\alpha_N(F)+\alpha_M((e^{X_N} K)(F))]}\;
w[(e^{X_N}\; e^{X_M} K)(F)]\Omega
\ee
Iterating with $U_j=U(M_j),\;X_j=X_{M_j},\; \alpha_j=\alpha_{M_j},\;
\; j=1,..,N$
\be \label{3.33}
U_1\; .. \; U_N\; w[F]\Omega
=e^{-i[\alpha_N(F)+\alpha_{N-1}((e^{X_N} K)(F))+..
+\alpha_1((e^{X_N}\;..\;e^{X_2}\;K)(F))]}\;
w[(e^{X_N}\; ..\; e^{X_1}\; K)(F)]\;\Omega
\ee
It follows that 
\ba \label{3.34}
&& U(sM)\;U(tN)\; U(-sM)\; U(-tN)\; w[F]\Omega
\nonumber\\
&=& e^{-i[\alpha_{-tN}(F)+\alpha_{-sM}((e^{-t X_N} K)(F))
+\alpha_{tN}((e^{-tX_N} e^{-sX_M}\; K)(F))
+\alpha_{sM}((e^{-tX_N} \; e^{-s X_M}\; e^{t X_N}\;K)(F))]}
\nonumber\\
&& w[(e^{-t X_N}\; e^{-s X_M} \; e^{t X_N}\; e^{s X_M} \;K)(F)]\Omega
\nonumber\\
&=&
e^{-i[\alpha^t_{-N}(F)+\alpha^s_{-M}((e^{-tX_N} K)(F))
+\alpha^t_{N}((e^{-tX_N} e^{-sX_M}\; K)(F))
+\alpha^s_{M}((e^{-t X_N} \; e^{-s X_M}\; e^{t X_N}\;K)(F))]}
\nonumber\\
&& w[F+st \;[X_N,X_M]\; K)(F)+O(s^3,s^2t,s t^2, t^3)]\;\Omega
\ea
We also expand the phase in (\ref{3.34}) keeping terms up to second order 
and note that the flows only need to be expanded to linear order as 
the integrals are already of first order
\ba \label{3.34a}
&&
\int_0^t\; dr\;[e^{-s X_M}\; e^{-t X_N}\; e^{r X_N}-e^{-r X_N}]\; V_N
+\int_0^t\; dr\;[e^{t X_N}\; e^{-s X_M}\; e^{-t X_N}\; e^{r X_M}-
e^{-t X_N}\; e^{-r X_M}]\; V_M
\nonumber\\
&=& st\; (X_N V_M-X_M V_N)+O(s^3,s^2t,s t^2, t^3)]
\ea
We now write (\ref{3.29} in the explicit form 
\ba \label{3.35}
C_M & =& M^\alpha \; [A_I \; G_\alpha^I(E) + V_\alpha(E)],\;
\nonumber\\
\{C_M,C_N\} &=& M^\alpha\;N^\beta\; 
[A_I\; G_\gamma^I(E)+V_\gamma(E)]
\kappa_{\alpha,\beta}\;^\gamma(E)
\\
&=:& A_I\; H_{M,N}^I(E)+L_{M,N}(E)
=\{C_M-V_M,C_N-V_N\}+\{C_M-V_M,V_N\}-\{C_N-V_N,V_M\}
\nonumber
\ea
which is again at most linear in $A_I$ and where we have ordered 
all dependence on $E$ to the right.
Thus replacing $G_M^I,\;V_M$ by 
$H_{M,N}^I,\; L_{M,N}$ and  
following the same steps as for $C_M$ we define 
\be \label{3.36}
U([M,N])\;w[F]\Omega:=\exp(i \{C_M,C_N\})\;w[F]\Omega=
e^{i\alpha_{M,N}(F)}\;w[(e^{-X_{M,N}} K)(F)]\;\Omega 
\ee
where $X_{M,N}$ is the Hamiltonian vector field of $A_I H_{M,N}^I$ and 
$\alpha_{M,N}(F)$ is the phase 
\be \label{3.36a}
\int_0^1\; dr\; e^{r\; X_{M,N}}\; L_{M,N}
\ee
Since 
\be \label{3.36b}
X_{sM,tN}=st\; X_{M,N}
=st\; X_{\{C_M-V_M,C_N-V_N\}}
=st\; [X_{C_M-V_M},X_{C_N-V_N}]
=st\; [X_M,X_N]
\ee
and 
\be \label{3.36c}
\alpha_{sM,sN}(F)=\alpha^{st}_{M,N}(F)
=st\; L_{M,N}(F)+O(s^3,s^2t,st^2,t^3)
=st\; (X_M V_N-X_N V_M)(F)+O(s^3,s^2t,st^2,t^3)
\ee
Now due to our choice of representation $E^I=i\partial/\partial A_I$ we have 
$[A_I,E^J]=-i\delta_I^J=-i\{A_I,E^J\}$ and thus expect 
$[C_M,C_N]=-i\{C_M,C_N\}=i\{C_M,C_N\}$ 
to leading order in the quantum theory i.e. to leading order
\be \label{3.36d}
U(M)\;U(N)\;U(-M)\;U(-N)=1-[C_M,C_N]=1+i\{C_M,C_N\}=U([M,N])
\ee
These relations 
establish that the composition law (\ref{3.32}) has resulted 
in the leading order exponentiated substitute (\ref{3.36d}) 
for the infinitesimal version (\ref{3.29}),
i.e. that the representation of the $U(M)$ is free of anomalies in the sense 
of the subsequent definition.
\end{proof}
The precise statement of anomaly freeness for the exponentiated versions 
is given in the following definition.
\begin{Definition} \label{def3.2} ~\\
Suppose that operators $U(M)=\exp(-i C_M),\; 
U([[M,N])=U(i\{C_M,C_N\})$ 
are defined on the common, dense, invariant domain given by 
the linear span of the $w[F]\Omega$ and suppose that 
\be \label{3.37}
U(sM)\;U(tN)\;U(-sM)\;U(-tN)\; w[F]\Omega=
e^{-i\alpha_{s,t}}\;w[F_{s,t}]\Omega,\;\;
U([sM,tN])\;w[F]\Omega=e^{-i\alpha'_{s,t}}\;w[F'_{s,t}]\Omega
\ee
Then the $U(M)$ are said to be represented free 
of anomalies on the Hilbert space $\cal H$ with dense span of the 
$w[F]\Omega$ iff $U(0)=U([0,0])=1_{{\cal H}}$ and 
\be \label{3.38}
[\frac{d}{ds}\;\frac{d}{dt}\;\{F_{s,t}-F'_{s,t})\}]_{s,t=0}=0
=[\frac{d}{ds}\;\frac{d}{dt}\;\{\alpha_{s,t}-\alpha'_{s,t})\}]_{s,t=0}
\ee
\end{Definition}
This definition is general enough to encompass the situation 
that the 1-parameter groups $s\mapsto U(sM)$ are not weakly continuous:
Thus, while we cannot take the derivatives or even limits of the 
$U(sM)$ we can take derivatives or limits of the $F_{s,t}$. 
Note also that we insist on independent quantisations of 
$U(M)=\exp(-i C_M),\; U([M,N])=\exp(-i\{C_M,C_N\})$ as otherwise 
we can trivially obtain anomaly freeness by declaring 
$U([M,N])=U(M) \; U(N)\;U(-M)\; U(-N)$. 

An expected but unusual property of the operators $U(M)$ is that 
their products $U(M_1)\; U(M_2)$ in general cannot be written in the form 
$U(M_3)$,
not even when $M_1,M_2$ are close to zero and thus $U(M_1), U(M_2)$ are
``close'' to ${\rm id}_{{\cal H}}$. This is precisely due to the fact that 
the $[X_{M_1},X_{M_2}]$ is not a linear combination with structure 
constants of the $X_{M_3}$ but with structure functions,
i.e. that we have a Lie algebroid rather than a Lie algebra structure.\\
\\
We now provide the details about the concrete situation in the 
U(1)$^3$ model. We treat only the case of zero cosmological constant,
its inclusion is straightforward given the general theory above.
First we integrate 
the smeared constraints by parts and write them in the form
\ba \label{3.9}
&& D[u]=<A,G_u>,\; 
C_\delta[M]=<A,G_{\delta,M}>,\;
\nonumber\\
&&
(G_u)^a_j(x)=-2[\partial_b(E^{[a}_j u^{b]})](x),\; 
(G_{\delta,M})_a^j(x)=-2[\partial_b(M \epsilon^{jkl} E^a_k E^b_l\; 
|\det(E)|^{(\delta-2)/2})](x)
\ea
thereby displaying the functions $G_M^I$ of theorem \ref{th3.1}
where the index $I=(a,j,x)$ is compound and summation over $I$ means 
summing over $a,j$ and integrating over $x$. Correspondingly,  
the $h_I$ of the theorem are given by 
\be \label{3.10}
[h_a^j(x,\kappa)]^b_k(y) = \delta_a^b\; \delta^j_k\; \delta_\kappa(x,y)
\ee
where $\kappa\mapsto \delta_\kappa$ is a mollified $\delta$ distribution,
i.e. a 1-parameter family of smooth functions converging to the $\delta$ 
distribution on Schwartz space over $\sigma$. With their help we define
the analog of the discrete derivative 
$\Delta_{\epsilon,I}$ on functionals of the functions $F^a_j(x)$ 
of the theorem by 
\be \label{3.10a}
(\Delta_{\epsilon,\kappa,a,j,x} W)[F]:=
\frac{1}{2\epsilon}
\{W[F+\epsilon h_a^j(x,\kappa)]-W[F-\epsilon h_a^j(x,\kappa)]
\ee
With all the $E^a_j(x)$ ordered to the outmost right, the assumption 
$G_u \Omega=G_M \Omega=0$ is met if we set for $\delta<2$
\be \label{3.10b}
[|\det(E)|^{(\delta-2)/2} E^a_j\; E^b_k](x)\; \Omega
:=\lim_{s\to 0+} 
[(s+|\det(E)|^{(2-\delta)/2})^{-1} E^a_j\; E^b_k](x)\; \Omega
=0
\ee
Then the theorem applies, with the understanding that the limit 
$\epsilon\to 0$ at the level of the functions $F$ is accompanied 
by the limit $\kappa\to 0$. 

As a result we obtain the {\it hypersurface deformation groupoid (HDG)} 
with explicit action on the common dense and invariant domain 
$\cal D$ given by the span of the $w[F]\Omega$ 
\ba \label{3.38}
U(u) &:=& e^{-i\; D[u]},\;\;U(M):=e^{-i\; C_\delta[M]},\;\;
\nonumber\\
U(u) \; w[F]\; \Omega &=& w[(e^{X_u}\; K)(F)]\; \Omega
\nonumber\\
U(M) \; w[F]\; \Omega &=& w[(e^{X_M}\; K)(F)]\; \Omega
\ea
with $X_u, \; X_M$ the Hamiltonian vector field of $D[u], C_\delta[M]$
respectively.

The action of $U(u)$ in (\ref{3.38}) is in fact the same as in LQG 
which works even in the non-Abelian setting because $D[u]$ is linear in 
$A$. As $D[u]$ is also linear in $E$, the flow of $D[u]$ preserves the 
linearity in $F$. For this reason, $e^{s\;D[u]} K$ can be worked out 
in closed form
\be \label{3.39}
[(e^{s\; X_u}\; K)(F)]^a_j(x)=[e^{s L_u}\; F]^a_j(x),\;
[L_u F]^a_j(x)=
[(u^b\; F^a_j)_{,b}-u^a_{,b}\; F^b_j](x)
\ee
where $L_u$ is the Lie derivative acting on vector field 
densities of weight one. One can 
check the implication 
\be \label{3.40}
\partial_a F^a_j=0 \;\;\Rightarrow\;\; \partial_a (L_u \; F^a_j)=0
\ee
which means that solutions $w[F]\Omega$ of the Gauss constraint are mapped 
to solutions of the Gauss constraint. Furthermore the flow preserves  
the space of vector field densities. 

As far as $C_\delta(M)$ is concerned, the fact that $D(u)$ generates 
spatial diffeomorphisms and the phase space dependent integrand 
$C_\delta(x);\; C[M]=\int\; d^3x\; M(x) \; C_\delta(x)$ is a scalar
density of weight $\delta$ implies 
\be \label{3.41}
\{D(u),C_\delta(M)\}=-C_\delta(L_u M),\;\;
L_u M=u^a M_{,a}-(\delta-1)\; u^a_{,a}\; M
\ee
i.e. $M$ acquires the geometrical intepretation of a scalar density 
of weight $-(\delta-1)$. Acordingly the net weight of 
$M_\delta(E):=M |\det(E)|^{(\delta-2)/2}$ is always $-1$. It follows
\be \label{3.42}
\{C_\delta(M),E^a_j(x)\}=
2\;[(\epsilon_{jkl}\; E^b_k\; E^a_l\;M_\delta(E))_{,b}](x)
\ee
whence 
\be \label{3.43}
[(X_M\; K)^a_j(F)](x)
=2\;[(\epsilon_{jkl}\; F^b_k\; F^a_l\;M_\delta(F))_{,b}](x)
\ee
Obviously (\ref{3.43}) is no longer linear in $F$, polynomial only for 
$\delta=2$. However, for any $\delta$ (\ref{3.43}) is divergence free 
because $\epsilon_{jkl} F^b_k F^a_l$ is antisymmetric in in $a,b$. Thus
the flow $e^{s\; X_M}$ preserves the space of solutions to the Gauss 
constraint.

To illustrate the degree of comlexity of the flow, let us work out the 
first few orders for the simplest (polynomial) case $\delta=2$. 
We set 
\be \label{3.44}
[B_M(F,G)]^a_j:=2\epsilon_{jkl} \; [M\; (F^{[b}_k\; G^{a]}_l]_{,b}
\;\;\Rightarrow\;\;
B_M(F,G)=B_M(G,F),\;B_M(F,F)=(X_M\; K)(F),\; \partial_a [B_M(F,G)]^a_j=0
\ee
which maps a pair of triples of divergence free 
vector densities to another triple of 
divergence free vector densities. It is also symmetric in $F,G$.

As an aside, note that 
\be \label{3.45}
[B_M(F,G)]^a_j=-\epsilon^{abc}\; \partial_b\; (\omega_M)_c^j,\;
(\omega_M)_c^j=M\epsilon^{jkl} \epsilon_{cde} F^d_k G^e_l
\ee
Thus if $M$ transforms as a density of weight $-1$ and 
$F,G$ as vector densities of weight $+1$ then $\omega_M^j$ is 
a 1-form. 

We have 
\ba \label{3.47}
X_M\; K &=& B_M(K,K)
\nonumber\\
X_M^2\; K &=& 
B_M(X_M\; K,K)+B_M(K,X_M\; K)=2\;B_M(K,B_M(K,K))
\nonumber\\
X_M^3\; K &=& 
2\; B_M(X_M\; K, B_M(K,K))+2\;B_M(K, X_M B_M(K,K)) 
\nonumber\\
&=& 2\; B_M(B_M(K,K), B_M(K,K))+4\;B_M(K, B_M(K,B_M(K,K)))
\ea
In general $X_M^n$ is a nested polynomial of order $n+1$ in $K$ involving
$n$ bilinear forms $B_M$ in all possible ways. It may well be possible 
to find the recursion relation for the numerical 
coefficients among these possible terms but we will not need them for what 
follows. \\
\\
We end this section with the remark that for non-integer density weight 
or density weight smaller than two, the space $\cal F$ to which the 
$F$ belong is restricted to non-degenerate elements, that is 
$\det(F)\not=0$ anywhere, the quantum trace of the classical condition 
that the classical metric be non-degenrate. This will be further analysed 
in the next section.

\section{Dual representation}
\label{s4}

The unitary operators $U[u],\;U[M]$ are defined densely on the 
span $\cal D$ of the $w[F]\Omega$ (finite linear combinations of those), 
with $F$ non-degenerate, i.e. in $\cal H$ and not some space of 
distributions. They do not act weakly continuously there, thus their algebra
can only be compared to the exponentiated classical hypersurface 
deformation algebra and this is what we did in the previous section, thereby 
establishing anomaly freeness on $\cal H$ in the sense defined. 
One may obtain an infinitesimal 
action on a certain space $L\subset {\cal D}^\ast$ of distributions 
on $\cal D$ where ${\cal D}^\ast$ is the algebraic dual of $\cal D$ i.e.
all linear functionals without continuity conditions. A general 
element $l\in L$ maybe written
\be \label{4.1}
l=\sum_F\; l[F]\; <w[F]\Omega,.>_{{\cal H}}
\ee
and for any operator $A$ with dense and invariant domain $\cal D$ we define 
its dual $A'$ on $L$ by 
\be \label{4.2}
[A'\;l](w[F]\Omega):=l(Aw[F]\Omega)=\sum_{F'}\;l[F']\; 
<w[F']\Omega,\;A\;w[F]\Omega>
\ee
While the sum in (\ref{4.1}) is over uncountably many $F$, the condition that 
$A w[F] \Omega\in {\cal D}$ makes sure that (\ref{4.2}) is finite. We may 
thus construct $U'[u],\; U'[M], \; U'[M,N]$
\be \label{4.3}
[U'[u]\; l][w[F]\Omega]=l[(e^{X_u}\;K)(F)],\; 
[U'[M]\; l][w[F]\Omega]=l[(e^{X_M}\;K)(F)],\;
[U'[[M,N]]\; l][w[F]\Omega]=l[(e^{X_{M,N}}\;K)(F)]
\ee
where $X_u,\;X_M,\; X_{M,N}$ are the Hamiltonian vector fields of 
$D[u],\;C[M],\;\{C[M],C[N]\}$ respectively. 

For given $u,M$ we define the one parameter groups 
$s\mapsto U'[su],\;s\mapsto U'[sM]$ to be continuous 
in the $L,{\cal D}$ topology 
if 
$l[U[su]\psi],\;l[U[sM]\psi]$ is continuous in $s$ for all $l\in L,\psi\in 
{\cal D}$. By (\ref{4.3}) this is equivalent to the requirement that the 
coefficient functions $l$ are continuous on the chosen space ${\cal F}$
that we sum over. Consider $\cal F$ to be some standard space e.g. 
divergence free Schwartz 
functions so that the flows $e^{s\; X_u}, \; e^{s\; X_M}$ which formally 
involve spatial derivatives of arbirtraily high orders are well defined.
This is however not sufficient: Unless the density weight of $C_\delta$ is 
an integer larger than or equal to two, the functions $F$ must be 
everywhere regular, i.e. $\det(F)\not=0$ everywhere because the 
flow of the Hamiltonian constraint involves arbitrarily large negative powers 
of $|\det(F)|$. This is the condition of {\it quantum non-degeneracy} 
\cite{30}.

We now compute the infinitesimal generators 
\ba \label{4.4}
&& (-iD'[u]\; l)[F]:=\frac{d}{ds} l[(e^{s X_u} K)[F]=(X_u l)[F]=
\{D[u],l(K)\}_{K=F}
\nonumber\\
&& (-iC_\delta'[M]\; l)[F]:=\frac{d}{ds} l[(e^{s X_M} K)[F]=(X_M l)[F]
=\{C_\delta[M],l(K)\}_{K=F}
\ea
which are nothing but classical Poisson brackets with functions of $E$ only.
Therefore the the algebra of the $D'[u],D'[M]$ is precisely an 
anti-represention of the HDA $\mathfrak{h}$ since  
\be \label{4.5}
i\; (\{C_\delta(M),C_\delta(N)\}'\;l)[F]=
(\frac{d}{s}\;\frac{d}{dt} (U'([sM,tN]) l)[F])_{s=t=0}
=\{\{C_\delta(M),C_\delta(N)\},l(K)\}_{K=F}
\ee
and closes without anomalies irrespective of the density weight of the 
Hamiltonian constraint provided the functionals $l$ are restricted 
to smooth regular functions representing non-degenerate quantum metrics.

\section{Solutions to the constraints by groupoid averaging}
\label{s5}
 
A general strategy to solve quantum constraints is to use ``group averaging''
\cite{31}, that is, to construct a so-called anti-linear rigging map 
$\eta:\; {\cal D}\mapsto {\cal S}\subset {\cal D}^\ast$ that maps 
the common, dense invariant domain $\cal D$ of the constraints to a 
subspace $\cal S$ of algebraic distributions on $\cal D$ which is in the 
kernel of the dual to all constraints. That is 
\be \label{5.1}
(\eta \psi)[C_\alpha\psi']=0
\ee for all $\psi,\psi'\in {\cal D}$ and $\alpha$ is some index that labels
a complete set of constraints. In case that the constraints $C_\alpha$ are the 
self-adjoint generators of  
a Lie algebra $\mathfrak{l}$ we may pass to the corresponding unitary Lie 
group $\mathfrak{g}$ generated by composition of the 
$g=U[M]=\exp(-i \sum_\alpha M^\alpha C_\alpha),\; 
M^\alpha\in \mathbb{R}$ and if $\mathfrak{g}$ 
admits a left invariant, normalised Haar measure $\mu$ then we may 
set 
\be \label{5.2}
\eta \psi:=\int\; d\mu(g)\; <g\;\psi,.>_{{\cal H}}
\ee
which satisfied $(\eta\psi)[g\psi']=(\eta\psi)[\psi']$ by unitarity 
$g^\dagger=g^{-1}$ which may be 
considered as the exponentiated version of (\ref{5.1})). Furthermore,
\be \label{5.3}
<\eta \psi,\eta \psi'>_\eta:=(\eta \psi')[\psi]
\ee
defines an inner product on those solutions.

It appears that we have good chances to apply this fomalism given the theory 
of section \ref{s3} which provides us with unitarities $U[u],\; U[M]$ 
labelling the constraints. Unfortunately, the $U[u], U[M]$ do not 
generate a group but a groupoid. While for a group we have at least 
formally a composition law $U[M_1]\;U[M_2]=U[M_3]$ where $M_3$ is
in general an infinite Baker-Campbell-Hausdorff series in $M_1,M_2$ 
involving the structure constants of $\mathfrak{l}$, for a groupoid 
such a relation does not hold, ``words'' formed by taking products 
of the ``alphabet letters'' $U[M]$ are in general independent of each 
other. Thus, there can be no group structure, no Haar measure and no 
rigging map as above. 

Let $\cal A$ be the set of alphabet letters consisting of all $U[u], U[M]$
and let ${\cal W}_N$ be the set of words $w$ with $N$ letters of the 
form $w=a_1..a_N$ where $w$ is not reducible to a word with fewer letters,
by using the fact that the alphabet consists both $a,a^{-1}$.  
We may try to form a discrete sum 
\be \label{5.4}
\eta \psi=\psi+\sum_{N=1}^\infty \; \sum_{w\in {\cal W}_N}\; \omega(w)\;
<w\psi,.>
\ee
where we have included a ``weight'' function $\omega$. Then one may ask that 
\be \label{5.5}
(\eta \psi)[a\psi']=(\eta \psi)[\psi']\;\; \forall\;\; a\in {\cal A}
\ee
Since $w\in {\cal W}_N$ contains words $a\; w',\; w'\in  W_{N-1}$
we see that $a^\dagger w=a^{-1} w=w'$ can both increase and decrease
word length by one unit so (\ref{5.5}) does not lead to an immediate 
contradiction.
However, even in case that $\cal A$ has finitely many unitary letters 
$a_1,..,a_k$ which have the same finite order $a_l^n=1,\;l=1,..,k$ 
but otherwise are free (no other relations) there are relations for an
infinite number of words to check. We therefore consider this approach 
to the groupoid situation as impractical.

The idea well known in the literature \cite{37}
is to pass to equivalent constraints 
that do form an algebra. Given first class constraints $C_\alpha$ 
on a phase space 
with coordinates $(y_\alpha,x^\alpha;p_\mu,q^\nu)$ 
one may solve the constraints for 
the $y_\alpha$ and rewrite them in the form
\be \label{5.6}
\hat{C}_\alpha=y_\alpha+h_\alpha(x;p,q)
\ee
These constraints are strictly Abelian $\{\hat{C}_\alpha,\hat{C}_\beta\}=0$ and 
therefore can be subjected to group averaging. Among the caveats to this 
strategy is the fact that the function $h_\alpha$ in general has several 
branches unless the constraints $C_M$ involve the $y_M$ only linearly.
This caveat is actually absent for the U(1)$^3$ as in fact {\it all} 
momenta appear at most linearly in the constraints. 

In the U(1)$^3$ situation we may write the constraints just 
in terms of the curvatures $F_{ab}^j$ modulo a term proportional 
to the Gauss constraint which annhilates the states $w[F]\Omega$ as 
$F$ is divergence-free. Thus $F_{ab}^j$ only depends on the transversal 
parts $A_{a\perp}^j$ while the longitudinal parts 
$A_{a\parallel}=A_a^j-A_{a\perp}^j$ drop both from the states and the 
constraints. It is thus natural to select the four momenta 
$A_{a\perp}^j,\;j=1,2$ as the $y_\alpha$ and the two momenta 
$A_{a\perp}^3$ as the $p_\mu$ with corresponding conjugate 
configuration coordinates $E^{a\perp}_j, E^{a\perp}_3$ as the
$x^\alpha,q^\mu$ respectively. Precisely this description 
of the reduced phase space has been given in \cite{27} of which we provide 
some details in the next section.

Thus instead of the $U(M)=\exp(-i C_M)$ we consider the 
$\hat{U}(M)=\exp(-i \hat{C}_M),\hat{C}_M=\sum_\alpha M^\alpha \hat{C}_\alpha$
which inherit the action from section \ref{s3}
\be \label{5.7}
\hat{U}(M)\;w[F]\;\Omega=w[(e^{\hat{X}_M}\; K)(F)]\;\Omega
\ee
where $\hat{X}_M$ is the Hamiltonian vector field of $\hat{C}_M$ because 
the $\hat{C}_M$ is still linear in the momenta. 

Still we cannot just integrate over the $M^\alpha\in \mathbb{R}$
because the $\hat{U}(M)$ are weakly discontinuous unitarities which 
is why integrals with respect to $M$ of $<\psi,U(M)\psi'>$ would 
simply vanish as the matrix element is supported on Haar measure 
zero sets. We are thus forced to consider instead the discrete (i.e. 
summation) measure to construct the rigging map
\be \label{5.8}
\eta \psi=\sum_M\; <\hat{U}(M)\psi,.>_{{\cal H}}
\ee
similar to the LQG approach to averaging the spatial diffeomorphism group
\cite{38}. We have 
\be \label{5.9}
(\eta \psi)[\hat{U}(M')\psi']=\sum_M\; <\hat{U}(M-M')\psi,.>_{{\cal H}}
=(\eta \psi)[\psi']
\ee
where unitarity and Abelian nature of the $\hat{U}(M)$ was used.
 
For $\psi=w[F]\Omega,\; w[F]=\exp(-i<A,F>)$ this can 
be further detailed as follows: Using $\partial_a F^a_j=0$ we may 
split $F^a_j=F^{a\perp}_j$ into the components 
$R^\alpha=F^{a\perp}_j(x),;j=1,2$ and $S^\mu=F^{a\perp}_3$ 
and we split the coordinate functions $[K(E,A)]^a_j(x)=E^a_j(x)$ 
into the corresponding parts $T^\alpha, Q^\mu$ where it is understood 
that summation over $\alpha,\mu$ includes an integral over $x$. In 
this notation $w[F]=w[(R,S)]=\exp(-i<y,R>-i<p,S>)$. Then
\ba \label{5.10}
&&\hat{U}(M)\; w[F]\;\Omega 
=\hat{U}(M)\; w[(R,S)]\;\Omega 
=w[(e^{\hat{X}_M} K)(F)]\;\Omega 
\nonumber\\
&=& w[((e^{\hat{X}_M}\; T)(R,S), (e^{\hat{X}_M}\; Q)(R,S))]\;\Omega 
=w[(R+M, (e^{\hat{X}_M}\; Q)(R,S)]\;\Omega 
\ea
where $(\hat{X}_M T)(R,S)=(M+T)(R,S)=M+R$ was used ($M$ is the constant 
function) and $\hat{X}_M$ is the Hamiltonian vector field of $\hat{C}_M$. 
Thus for a general function 
\be \label{5.11}
\psi=\sum_F\; \psi(F)\; w[F]\;\Omega
=\sum_{R,S}F\; \psi(R,S)\; w[(R,S)]\;\Omega
\ee
with $\psi(F)\not=0$ for at most countably many $F$ we have 
\ba \label{5.12}
&&\sum_M\; \hat{U}(M)\; \psi
=\sum_{R,S,M}\; \psi(R,S)\; w[(R+M,(e^{\hat{X}_M}\; Q)(R,S))]\Omega
\nonumber\\
&=&\sum_{R,S,\tilde{M}}\; \psi(R,S)\; 
w[(\tilde{M},(e^{\hat{X}_{\tilde{M}-R}}\; Q)(R,S))]\Omega
\nonumber\\
&=&\sum_{R,S,\tilde{M}}\; \psi(R,S)\; 
w[(\tilde{M},(e^{\hat{X}_{\tilde{M}}}\; e^{\hat{X}_{-R}}\; Q)(R,S))]\Omega
\ea
where we have introduced a new summation variable $\tilde{M}=M+R$ 
in the second step and in the third we used that the $\hat{X}_M$ are 
Abelian.

Next we have by unitarity and due to the Abelian property
\ba \label{5.13}
&& <\sum_M \; \hat{U}(M) \psi,\; 
<\sum_{M'} \; \hat{U}(M') \psi'>_{{\cal H}}
=
\sum_{M,M'}\; <\hat{U}(M-M') \psi,\psi'> 
\nonumber\\
&=&
[\sum_{M'} \; 1]\;\; [\sum_{M}\; <\hat{U}(M) \psi,\psi'>] 
=:{\rm Vol}({\cal M})\; <\eta \psi,\eta\psi'>_\eta
\ea
where $\cal M$ is the space of $M$ that we sum over which is the 
same as the space ${\cal R}$ of $R$. Accordingly we have for the 
rigging inner product
\ba \label{5.14}
{\rm Vol}({\cal M})\; <\eta \psi,\eta\psi'>_\eta
&=& \sum_{R,S,M,R',S',M'} \; \psi^\ast(R,S)\;\psi'(R',S')
\delta_{M,M'}\; 
\delta_{(e^{X_M}\; e^{\hat{X}_{-R}} Q)(R,S),(e^{\hat{X}_{M'}}\; 
e^{\hat{X}_{-R'}} Q)(R',S')}
\nonumber\\
&=& {\rm Vol}({\cal M})\;
\sum_{R,S,R',S'} \; \psi^\ast(R,S)\;\psi'(R',S')\;
\delta_{([e^{\hat{X}_Z}\; Q]_{Z=-T})(R,S),([e^{\hat{X}_Z}\; Q]_{Z=-T})(R',S')}
\nonumber\\
&=& {\rm Vol}({\cal M})\;
\sum_{R,S,R',S'} \; \psi^\ast(R,S)\;\psi'(R',S')\;
\delta_{O_Q(R,S),O_Q(R',S')}
\nonumber\\
&=& {\rm Vol}({\cal M})\;\sum_O\;
[\sum_R\; \psi(R,O_Q^{-1}(R,O))]^\ast\;
[\sum_{R'}\; \psi'(R',O_Q^{-1}(R',O))]
\ea
where we used that $e^{\hat{X}_M}$ is invertible in the second step
and rewrote the arguments in the remaining Kronecker, in the third 
we introduced the notation
\be \label{5.15}
O_H:=[e^{\hat{X}_Z} \; H]_{Z=-T} 
\ee
which is the projection of the phase space function $H$ to the relational
gauge invariant observable corresponding to the gauge $T=0$ 
\cite{37} and in the last we solved the Kronecker using
$S=O_Q^{-1}(R,O)$ which is the inversion of $O_Q(R,S)=O$ at fixed $O$.

It follows that the physical Hilbert space obtained by the rigging 
method can be identified with the Hilbert space with dense span 
given by the $\sum_O\; \hat{\psi}(O)\; w[(0,O)]\Omega$ via the 
the unitary map $V:\eta \psi\mapsto \hat{\psi}$ with 
\be \label{5.16}
\hat{\psi}(O)=\sum_R\; \psi(R,O_Q^{-1}(R,O))
\ee
As expected, the physical states depend only on (relational) Dirac
observables $O$ corresponding to the configuration variables $q^\mu$ not 
subject to the gauge fixing $x^\alpha=0$ that defines these relational 
obeservables. As is well known \cite{37}, the phase space defined 
by the relational observables corresponding to $q^\mu,\;p_\mu$ and the 
gauge fixing condition $x^\alpha=0$ is completely equivalent to reduced 
phase space obtained by solving the constraints for $y_\alpha$ in the 
gauge $x^\alpha=0$ and keeping $q^\mu,p_\mu$ as ``true degrees of freedom''.
Whenever that gauge fixing is not complete but leaves a 1-parameter family
of residual gauge transformations, one may use the generator of those 
residual transformations as physical or reduced Hamiltonian. We thus 
turn to the reduced phase space description in the next section and 
compute and quantise the corresponding physical Hamiltonian.

\section{Reduced phase space, physical Hamiltonian and quantisation}
\label{s6}

The classical part of this section is a slightly generalised version of 
a part of \cite{27}. We thus will be brief and refer the reader to 
\cite{27} for details.\\
\\
The classical constraints can be written in density $\delta$ form 
\be \label{6.1}
D^\delta_j=F_{ab}^k\; E^a_j\;E^b_k\; |\det(E)|^{(\delta-2)/2},\;
C^\delta=F_{ab}^j\; \epsilon_{jkl} E^b_k,\; E^c_l\;|\det(E)|^{(\delta-2)/2}
\ee
with $F_{ab}^j=2\partial_{[a}\; A_{b]}^j$ where $G_j=\partial_a\; E^a_j=0$
was used so that $E^a_j=E^{a\perp}_j$ is already transversal. 
The gauge condition we wish to impose is on 
$E^{a\perp}_\alpha,\; \alpha=1,2$ and 
correspondinly we want to solve (\ref{6.1}) for $A_{a\perp}^\alpha$. 
We can use the Gauss constraint to install the three Coulomb gauges 
$A_{a\parallel}^j=A_a^j-A_{a\perp}^j=0,\; j=1,2,3$ as it was done in \cite{27}. 
Note that while $\partial_a E^a_j=0$ defining the solutions $E^{a\perp}_j$
does not require a background metric, the definition of $A_{a\perp}^j$ does 
require a background metric. We will pick one and proceed as in \cite{27}.

In this paper we slighly deviate from \cite{27} 
and just impose two Coulomb gauges
$A_{a\parallel}^\alpha=0,\;\alpha=1,2$ and also only solve two Gauss 
constraints $G_\alpha=\partial_a E^a_\alpha=0$. Thus we keep as true 
degrees of freedom the three canonical pairs $A_a^3,\; E^a_3$ and 
keep the Gauss constraint $G_3=\partial_a E^a_3$ still in place, to be dealt 
with later. We will see that after having reduced the six constraints 
$G_\alpha, D_a, C$ we obtain a theory with a reduced Hamiltonian $H$ 
and a constraint $G:=G_3$ under which $H$ is invariant and $H$ will 
be linear in $F_{ab}^3$ and non-linear in $E^a_3$ \cite{27}. The resulting 
theory is thus a special type of {\it non-linear, self-interacting, local 
electrodynamics}. \\
\\
Proceeding to the details, in this paper we consider just the case 
$\sigma=\mathbb{R}^3$. More general manifolds can be treated with 
adapted methods. We pick a global Cartesian coordinate system 
$x^a,\;a=1,2,3$ on $\sigma$.
The gauge condition on the solution $E^{a\perp}_\alpha$ of 
$\partial_a E^a_\alpha=0$ that we pick is (remember $\alpha=1,2$) 
\be \label{6.2}
E^a_\alpha=\delta^a_\alpha
\ee
To distinguish the coordinate directions from the frame directions, 
we write $x,y,z$ for $a=1,2,3$. Then (\ref{6.2}) is a compact notation for 
$E^x_1=E^y_2=1,\;E^x_2=E^y_1=E^z_1=E^z_2=0$. Obviously the Gauss constraints
$\partial_a E^a_\alpha=0$ are identically satisfied. The reason 
why we do not impose $E^a_\alpha=0$ is that we want to keep the model 
as close as possible to GR and thus insist on non-degenerate 
metrics, thus $\det(\{E^a_j\})=E^z_3\not=0$ is still possible. 

We need to show that the six conditions 
(\ref{6.2}) can be always installed no matter 
from which configuration of the $E^a_\alpha$ we start from and that 
the six constraints $G_\alpha, D_a, C$ can always 
be solved for. To do this 
we rewrite (\ref{6.1}) in terms of the density -1 inverse 
$E_a^j$ and the density +1 magnetic 
field $B^a_j$
\be \label{6.3}
\det(E)\; E_a^j=\frac{1}{2}\epsilon_{abc}\epsilon^{jkl} E^b_k E^c_l,\;\;
2\;B^a_j=\epsilon^{abc} F_{bc}^j
\ee
from which
\be \label{6.4}
D^0_j=\epsilon_{jkl} B^a_k\; E_a^l,\; C^0=B^a_j\; E_a^j
\ee
which have density weight zero.
As $E^a_j$ is non-degenerate we can decompose 
$B^a_1=u^j \; E^a_j,\;B^a_2=v^j \; E^a_j$ and find 
\ba \label{6.5}
D^0_1 &=& B^a_2 E_a^3-B^a_3 E_a^2=v^3-B^a_3 E_a^2
\nonumber\\
D^0_2 &=& B^a_3 E_a^1-B^a_1 E_a^3=B^a_3 E_a^1-u^3
\nonumber\\
D^0_3 &=& B^a_1 E_a^2-B^a_2 E_a^1=u^2-v^1
\nonumber\\
C^0 &=& B^a_1 E_a^1+B^a_2 E_a^2+B^a_3 E_a^3=u^1+v^2+B^a_3 E_a^3
\ea
which can be solved algebraically for $u^3,v^1,v^2,v^3$ 
thus 
\ba \label{6.6}
B^a_1 &=& u^1 \; E^a_1+ u^2\; E^a_2+ [B^b_3 E_b^1] E^a_3
\nonumber\\
B^a_2 &=& u^2 \; E^a_1+ (B^b_3 E_b^3-u^1)\; E^a_2+ [B^b_3 E_b^2] E^a_3
\ea
The coefficients $u^1,u^2$ are constrained by the Bianchi identities 
$\partial_a B^a_\alpha=0$
\ba \label{6.7}
&& \delta^{\alpha\beta} E^a_\alpha\; u_{\beta,a}=-[(B^b_3 E^b_1) E^a_3]_{,a} 
=:-t\nonumber\\
&& \epsilon^{\alpha\beta} E^a_\alpha\; u_{\beta,a}=
-[(B^b_3 E^b_3) E^a_2+(B^b_3 E^b_2) E^a_3]_{,a} 
=:-r
\ea
where $\partial_a E^a_\alpha=0$ was used and $\epsilon^{\alpha\beta}$ 
is the skew symbol in 2 dimensions. We introduce with $I,J,K,..\in 
\{x,y\}$ the 2-dimensional divergence and curl
\be \label{6.8}
d:=\delta^{\alpha\beta} E^I_\alpha\; u_{\beta,I},\;\;
c:=\epsilon^{\alpha\beta} E^I_\alpha\; u_{\beta,I},\;\;
\ee
then 
\be \label{6.9}
\delta^{\alpha\beta}\; E^z_\alpha\; u_{\beta,z}=-(t+d),\; 
\epsilon^{\alpha\beta}\; E^z_\alpha\; u_{\beta,z}=-(r+c),\; 
\ee
The two equations (\ref{6.9}) provide a quasi-linear (even linear)
first order PDE system in two functions $u_1,u_2$. By the 
Cauchy-Kowalewskaja (CK) theorem \cite{38} maximal analytic 
and unique solutions of (\ref{6.9}) exist for real analytic 
``initial data'' $u_\alpha^0$ on a surface transversal to the z coordinate 
lines (e.g. the surface $z=0$), real analytic inhomogeneities $t,d$ and 
real analytic $E^a_\alpha$ provided that the matrix 
\be \label{6.10}
\left( \begin{array}{cc}
E^z_1 & E^z_2 \\
-E^z_2 & E^z_1   
\end{array}
\right)
\ee
is non-degenerate i.e. $\delta^{\alpha\beta} E^z_\alpha E^z_\beta>0$
(non-characteristic condition). 
In that case we can solve (\ref{6.9}) for $u_{\alpha,z}$ and 
then can compute the Taylor expansion of $u_\alpha$ off $z=0$ 
by CK iteration of the PDE system. We can argue the same 
way by solving instead for the $x$ and $y$ derivatives. It is 
not possible that $\delta^{\alpha\beta} E^a_\alpha E^a\beta=0$ for 
more than one of  
$a=x,y,z$ because otherwise this would imply that say
$E^I_\alpha=0;\; I=x,y,; \alpha=1,2$
and the metric would be degenerate. Thus w.l.g. we may pick the 
$z$ direction as long as the non-characteristic condition above holds.

We may solve (\ref{6.9}) also in case  
that $E^z_1=E^z_2=0$ as long as the metric is non degenerate. For this requires 
that $\det(\{E^I_\alpha\}), E^z_3\not=0$ and 
$\partial_a E^a_\alpha=\partial_I E^I_\alpha=0$ ensures that there exist 
functions $e_\alpha$ with $E^I_\alpha=\epsilon^{IJ} e_{\alpha,J}$ by 
simple connectedness of $\sigma=\mathbb{R}^3$. It follows 
that $x,y,z\mapsto (\hat{x},\hat{y},\hat{z})=
(e_1(x,y,z),e_2(x,y,z),z)$ is a diffeomorphism. 
Furthermore, as (\ref{6.9}) vanishes identically we 
have with $u_\alpha(x,y,z)=:\hat{u}(e_1(x,y,z), e_2(x,y,z),z)$ 
\ba \label{6.11}
d &=& \delta^{\alpha\beta}\; \epsilon^{IJ}\; e_{\alpha,J}\; u_{\beta,I}
=\det(E^I_\alpha)\; \epsilon^{\alpha\beta} \partial_\alpha \hat{u}_\beta
=-t
\nonumber\\
c &=& \epsilon^{\alpha\beta}\; \epsilon^{IJ}\; e_{\alpha,J}\; u_{\beta,I}
=\det(E^I_\alpha)\; \delta^{\alpha\beta} \partial_\alpha \hat{u}_\beta
=-r
\ea
Switching to those coordinates and denoting by $\hat{t},\hat{r}$ the 
transformed functions $t/\det(E^I_\alpha),r/\det(E^I_\alpha)$ 
the solution is obtained as 
\be \label{6.12}
\hat{u}_\alpha=\Delta_2^{-1}[\epsilon^{\beta\alpha} \hat{t}_{,\beta}
+\hat{r}_{,\alpha}] +\epsilon^{\alpha\beta} \hat{h}_\beta
\ee
where $\Delta_2=\delta^{\alpha\beta}\partial_\alpha\partial_\beta$ and 
$h_\alpha$ is a homogeneous solution of (\ref{6.11}) 
\be \label{6.13}
\hat{h}_{1,1}-\hat{h}_{2,2}=\hat{h}_{1,2}+\hat{h}_{2,1}=0
\ee
i.e. a solution of the Cauchy-Riemann (CR) equations in $\hat{x},\hat{y}$
which means that $\hat{h}_1+i\hat{h}_2$ is a holomorphic function in 
$\hat{x}+i\hat{y}$. 

This shows that unique (up to the holomorphic function freedom
at characteristic surfaces) and maximal analytic solutions to the 
constraints $D_j=0, C=0$ in terms of $B^a_\alpha$ always exist if the 
$B^a_3, E^a_j$ are real analytic which can in principle be computed 
to arbitrary precision using Taylor expansion. Solutions may also 
exist outside the analytic category but this will not be of relevance for
what follows.\\
\\
\\
To see whether the gauge $E^a_\alpha=\delta^a_\alpha$ can be installed
we consider a general gauge transformation of $E^a_\alpha$
\ba \label{6.14}
\delta E^a_\alpha(x)
&=&
\{\int\;d^3y\;[\xi^j\; D^0_j+N\; C^0\delta](y),E^a_\alpha(x)\}
=\{\int\;d^3y\;B^b_j[\epsilon_{jkl}\;E_b^k \xi^l+N\; E_b^j](y),
E^a_\alpha(x)\}
\nonumber\\
&=& \epsilon^{abc}\;\partial_b[\epsilon_{\alpha kl} E_c^k \xi^l+
N\delta_\alpha^j E_c^j](x)
\ea
which shifts $E^a_\alpha$ by a divergence free vector density thus 
preserving the Gauss constraint $\partial_a E^a_\alpha=0$. If we work 
with density weight $\delta$ constraints then the relation between 
$\xi^j,N$ in (\ref{6.14}) and the density weight zero shift and 
lapse functions $n^a,n$ is 
\be \label{6.14a}
n^a\; |\det(E)|^{(\delta-1)/2}=\xi^j\; E^a_j\;|\det(E)|^{-1},\;
n\; |\det(E)|^{(\delta-2)/2}=\frac{N}{2}\;|\det(E)|^{-1},\;
\ee
To install 
$E^a_\alpha=\delta^a_\alpha$ we must have $\delta E^a_\alpha = 
E^a_\alpha -\delta^a_\alpha=\Delta^a_\alpha$ with 
$\partial_a \Delta^a_\alpha=0$. Since $\mathbb{R}^3$ is simply connected 
we find $w_{a\alpha}$ with $\Delta^a_\alpha=\epsilon^{abc} \partial_b 
w_{c\alpha}$. We could impose w.l.g. 
$\delta^{ab} \partial_a w_{b\alpha}=0$ but this will not be needed. It follows 
that (\ref{6.14}) is solved by 
\be \label{6.15}
\epsilon_{\alpha kl} E_a^k \xi^l+N\delta_\alpha^j E_a^j
=w_{a\alpha}+\partial_a g_\alpha
\ee
where $g_\alpha$ are free functions. These are six equations for six 
free functions $\xi^j, N, g_\alpha$ and we can solve them as follows:\\
Contracting (\ref{6.15}) by $E^a_j$ we obtain the equivalent system
\be \label{6.16}
\epsilon_{\alpha jk} \xi^k+N\delta_\alpha^j 
=E^a_j[w_{a\alpha}+\partial_a g_\alpha]=:Z^j_\alpha
\ee
We first solve algebraically for $\xi^j, N$. For $j=3$ and $j=\beta$ 
respectively
\ba \label{6.17}
&&
-\epsilon_{\alpha\beta} \xi^\beta=Z^3_\alpha
\nonumber\\
&& \epsilon_{\alpha\beta} \xi_3+N\delta_{\alpha\beta} 
Z_{\beta\alpha}\equiv Z^\beta\;\alpha
\ea
which yields
\be \label{6.18}
\xi^\alpha=\epsilon^{\alpha\beta} Z^3_\beta,\; 
\xi_3=-\frac{1}{2}\epsilon^{\alpha\beta} \; Z_{\alpha\beta},\;
N=-\frac{1}{2}\delta^{\alpha\beta} \; Z_{\alpha\beta},\;
\ee
and that $Z_{\alpha\beta}$ does not have a trace free symmetric part
\be \label{6.19}
Z_{(\alpha\beta)}-\frac{1}{2}\delta_{\alpha\beta}\delta^{\gamma\delta}
Z_{\gamma\delta}=0
\ee
The system (\ref{6.19}) only involves $g_\alpha$. Its independent 
ingredients are only two equations 
\be \label{6.20}
Z_{12}+Z_{21}=0,\; Z_{11}-Z_{22}=0
\ee
which maybe written as 
\be \label{6.21}
E^a_1[w_{a2}+\partial_a g_2]+E^a_2[w_{a1}+\partial_a g_1]=0,\;
E^a_1[w_{a1}+\partial_a g_1]-E^a_2[w_{a2}+\partial_a g_2]=0
\ee
With the relabelling 
\be \label{6.22}
u_1:=g_2,\; w_2:=g_1,
t:=E^a_1 w_{a2}+E^a_2 w_{a1}
r:=E^a_1 w_{a1}-E^a_2 w_{a2}
\ee
the system (\ref{6.21}) turns into exactly the system (\ref{6.7}) which 
therefore has a unique maximal analytic solution $g_\alpha$ given suitable 
initial conditions. Having obtained such a solution $g_\alpha$ 
the gauge is installed by using that $g_\alpha$ in (\ref{6.18}).

Once the gauge is installed, we ask what residual gauge transformations 
leave $E^a_\alpha=\delta^a_\alpha$ invariant. Then we must solve 
(\ref{6.21}) with $w_{a\alpha}=0$ and $E^a_\alpha=\delta^a_\alpha$
\be \label{6.23}
g_{2,x}+g_{1,y}=0,\; g_{1,x}-g_{2,y}=0
\ee
i.e. $g=g_1+i g_2$ is a holomorphic function of $x+i y$ and its $z$ dependence
is not constrained. Then (\ref{6.18}) becomes with 
$Z^3_\alpha=E^a_3 g_{\alpha,a},\;Z_{\beta\alpha}=\partial_\alpha g_\beta$
\be \label{6.24}
\xi^\alpha=\epsilon^{\alpha\beta} E^a_3\; \partial_a g_\alpha,\; 
\xi_3=-\frac{1}{2}\epsilon^{\alpha\beta} \; \partial_\alpha g_\beta,\;
N=-\frac{1}{2}\delta^{\alpha\beta} \; \partial_\alpha g_\beta,\;
\ee
%E^a_j[w_{a\alpha}+\partial_a g_\alpha]=:Z^j_\alpha
%w_{a\alpha}=0
%Z^3_\alpha=E^a_3 g_{\alpha,a}, Z^\beta_\alpha=g_{\alpha,\beta}   
%
The relation between $\xi^j,N$ and the shift and lapse fields $n^a,n$
with density
weight zero is given in (\ref{6.14a}) for the density $\delta$ form 
of the constraints. In the given gauge $\det(E)=E^z_3$ and  
$|\det(E)|^{(1+\delta)/2}\;n^a=
E^a_3 \xi^3+\delta^a_\alpha \xi^\alpha,\;
|\det(E)|^{\delta/2}\;n=\frac{N}{2}$. If we want 
to keep the relation with GR intact then we want that $n^a$ decays at 
infinity while $n$ approaches a constant $c$
as to obtain the Minkowski metric 
asymptotically with Minkowski/Euclidian
time coordinate $x^0=ct$ where $c>0$ is some 
parameter which has the value of the speed of light when $t$ is the 
usual asymptotic time but we allow variable $c$ corresponding to 
1-parameter reparametrisation freedom $t\mapsto t'=ct/c'$. 
Furthermore 
\be \label{6.25}
Q^{ab}-\delta^{ab}=\delta^{jk} E^a_j E^b_k-\delta^{ab}
=E^a_3 E^b_3+\delta^a_1 \delta^b_1+\delta^a_2\delta^b_2-\delta^{ab}
=E^a_3 E^b_3-\delta^a_3 \delta^b_3
\ee
also decays, i.e. $E^z_\alpha$ and $E^z_3-1$ decay. Thus 
$\xi^j$ must decay while $N$ approaches $2c$. Since 
$-2\xi_3=g_{2,x}-g_{1,y}=2g_{2,x}=-2g_{1,y}$ and 
$\xi_\alpha=\epsilon^{\alpha\beta} E^a_3 g_{\beta,a}$ the decay 
of $\xi_j$ can be ensured by demanding $g_{\alpha,z}=0,\;
g_{2,x}=g_{1,y}=0$ which means that $g_\alpha$ just depends on $x^\alpha$
and by holomorphicity $g_{1,x}=g_{2,y}=-c$ is a constant. Then $N=2c$ 
and $n$ approaches $c$ asymptotically indeed. Moreover, 
$\xi^3=0, \xi^\alpha=-c\epsilon^{\alpha\beta} \delta_{a\beta} E^a_3$.
As shown in \cite{27}
the analysis can be completed by providing suitable decay behaviour 
of all fields in order to make all constraints functionally differentiable 
and finite.

Thus, as intended we obtain a 1-parameter freedom of residual gauge 
transformations which give rise to a reduced Hamiltonian $H$. This
is defined as follows. Let $G$ be any functional of the true degrees of 
freedom $E^a_3, A_a^3$ and let $n_c, n^a_c$ be the 1-parameter family 
of lapse and shift functions just found. Let 
$\bar{E}^a_\alpha=\delta^a_\alpha$ be the constrained/gauge fixed 
pure gauge configuration degrees of freedom and  
$\bar{A}_a^\alpha$ be the constrained/gauge fixed 
pure gauge momentum degrees of freedom whose existence we have demonstrated 
above but whose explicit expression will not be needed in what follows.
Then  
\be \label{6.26}
\{H^\delta,G\}=\{\vec{D}^\delta(\vec{n})
+C^\delta(n),G\}_{n=n_c,\vec{n}=\vec{n}_c,
E^a_\alpha=\bar{E}^a_\alpha,A_a^\alpha=\bar{A}_a^\alpha}
\ee
where $D^\delta,C^\delta$ are taken in the density weight $\delta$ form.
The Hamilltonian determined by (\ref{6.26}) has been computed in
\cite{27} for $\delta=1$ 
using an abstract argument valid for general theories at most 
linear in the momenta. To make exposition self-contained, we provide here an 
alternative argument which uses the concrete structure of our constraints.
We expand 
\be \label{6.27}
C^\delta_\mu(n^\mu):=\vec{D}^\delta(\vec{n})+C^\delta(n)
=\int\; d^3x\; n^\mu[B^a_\alpha\; h_{a\mu}^{\delta,\alpha}+h^\delta_\mu]
\ee
where $n^0=n$ and $h_{a\mu}^{\delta,\alpha}$ depends only on $E^a_j$ while 
$h_\mu^\delta$ depends linearly on $B^a_3$ and on $E^a_j$. We denote 
by $\bar{B}^a_\alpha$ the solution of (\ref{6.27}) at 
$E^a_\alpha=\bar{E}^a_\alpha$ and denote the values of 
$h_{a\mu}^{\delta,\alpha},\;h_\mu$ by
$\bar{h}_{a\mu}^{\delta,\alpha},\;\bar{h}_\mu$ resepectively. Moreover,
the solution of the gauge stability 
\be \label{6.28}
\{C^\delta_\mu(n^\mu),E^a_\alpha\}=0=\epsilon^{abc}\partial_b[ n^\mu\; 
h_{c\mu}^{\delta,\alpha}\;\;\Rightarrow\;\;n^\mu h_{a,\mu}^{\delta,\alpha}]
=\partial_a g^{\delta,\alpha}
\ee
at $E^a_\alpha=\bar{E}^a_\alpha$ is denoted by $\bar{n}^\mu,\; 
\bar{g}_\alpha^\delta$. Thus we have
\be \label{6.29}
\bar{B}^a_\alpha \; \bar{h}_{a\mu}^{\delta,\alpha}+\bar{h}_\mu=0,\;\;
\bar{n}^\mu \;\bar{h}_{a,\mu}^{\delta,\alpha}
=\partial_a \bar{g}^{\delta,\alpha} 
\ee
Thus 
\ba \label{6.30}
\{H^\delta,G\} 
&=& 
\int\; d^3x\;\bar{n}^\mu\;
[\bar{B}^a_\alpha\;
\{h_{a\mu}^{\delta,\alpha},G\}_{E^a_\alpha=\bar{E}^a_\alpha}
+\{h_\mu^\delta,G\}_{E^a_\alpha=\bar{E}^a_\alpha}]
\nonumber\\
&=& 
\int\; d^3x\;\bar{n}^\mu\;
[\bar{B}^a_\alpha\;
\{\bar{h}_{a\mu}^{\delta,\alpha},G\}
+\{\bar{h}_\mu^\delta,G\}]
\nonumber\\
&=& 
\int\; d^3x\;
[\bar{B}^a_\alpha\;(
\{\bar{n}^\mu\;\bar{h}_{a\mu}^{\delta,\alpha},G\}
-\bar{h}_{a\mu}^{\delta,\alpha}\;\{\bar{n}^\mu,G\})
+\bar{n}^\mu\;
\{\bar{h}_\mu^\delta,G\}]
\nonumber\\
&=& 
\int\; d^3x\;
[\bar{B}^a_\alpha\;
\partial_a\{\bar{g}_\alpha^\delta,G\}
+\bar{h}_{\mu}^\delta\;\{\bar{n}^\mu,G\})
+\bar{n}^\mu\;
\{\bar{h}_\mu^\delta,G\}]
\nonumber\\
&=& 
\int\; d^3x\;\{\bar{n}^\mu\;\bar{h}^\delta_\mu,G\}]
\ea
where in the first step we used that the Poisson bracket must involve the 
constraint  and that $\{B^a_\alpha,G\}=0$, in the second we used 
that while $h_{a\mu}^{\alpha,\delta},h_\mu$ depend on $E^a_\alpha$
before fixing to $\bar{E}^a_\alpha$ the Poisson bracket with $G$ only 
cares about the $E^a_3$ dependence, in the third we used the Leiniz rule,
in the fourth we used both identities (\ref{6.29}), in the fifth we 
integrated by parts (the decay behaviour of the fields excludes 
a boundary term \cite{27}), used the Bianchi identity on the magnetic 
field and applied the Leibniz rule again. It follows that
\ba \label{6.31}
H^\delta &=& \int\;d^3x\; \bar{n}^\mu\;\bar{h}_\mu
\nonumber\\
&=& \int\; d^3x\; [\bar{\xi}_j\epsilon^{j3l} B^a_3 \bar{E}_{al}+
\bar{N}\;B^a_3\; \bar{E}_a^3]   
\nonumber\\
&=& \int\; d^3x\; [\bar{\xi}_\alpha\epsilon^{\alpha 3l} B^a_3 \bar{E}_{al}+
\bar{N}\;B^a_3\; \bar{E}_a^3]   
\nonumber\\
&=& \int\; d^3x\; [\bar{\xi}_\alpha\epsilon^{\beta\alpha} B^a_3 
\bar{E}_{a\beta}+\bar{N}\;B^a_3\; \bar{E}_a^3]   
\nonumber\\
&=& c\int\; d^3x\; 
[-\delta_b^\alpha E^b_3\;B^a_3 \bar{E}_{a\alpha}+2\;B^a_3\; \bar{E}_a^3]   
\nonumber\\
&=& c\int\; d^3x\;[E^z_3]^{-1}\; 
[-\delta_b^\alpha E^b_3\;B^a_3 (\delta_a^\alpha E^z_3
-\delta_a^3\;E^\alpha_3)+2\;B^z_3]
\nonumber\\
&=& c\int\; d^3x\;[E^z_3]^{-1}\; 
[-\delta_{IJ} B^I_3\; E^J_3\;E^z_3+B^z_3\;(2+\delta_{IJ} E^I_3\; E^J_3)] 
\nonumber\\
&=:& <B,f(E)>
\ea 
where we used 
\be \label{6.32}
|\det(\bar{E})| \bar{E}_a^3
%=\epsilon_{abc} E^b_1 E^c_2
=\delta_a^z,\;
|\det(\bar{E})| \bar{E}_a^\alpha
%=\frac{1}{2}\epsilon_{\alpha kl}
%\epsilon_{abc} E^b_k E^c_l=
%=\epsilon_{\alpha\beta}
%\epsilon_{abc} E^b_\beta E^c_3=
%=\epsilon_{\alpha\beta}
%\epsilon_{a\beta c} E^c_3=
%=\epsilon_{\alpha\beta}
%[\delta_a^3\;\epsilon_{\beta\gamma} E^\gamma_3=
%+\delta_a^\gamma\epsilon_{\gamma\beta} E^z_3]
=[-\delta_a^3\;E^\alpha_3=
+\delta_a^\alpha E^z_3]
\ee
and the notation 
\be \label{6.32a}
f_a(E)=\frac{1}{E^z_3}[\delta_a^z(2+\delta_{IJ} E^I_3 E^J_3)
-\delta_a^I \delta_{IJ} E^J_3\; E^z_3]
\ee
Note that $H=H^\delta$ is independent of the choice of $\delta$. This 
is true in general: The reduced Hamiltonian does not care about the 
density weight which was used for the constraint because one can redefine 
lapse and shift functions to unity density weight and finds unique 
values for those when gauge fixing.\\
\\
It follows that the reduced theory is a kind of non-linear electrodynamics
with gauge invariant electric and magnetic fields $E^a:=E^a_3, 
B^a:=B^a_3=\epsilon^{abc} \partial_b A_c, A_a=A_a^3$ which is subject to
the Gauss constraint $G=\partial_a E^a=G_3$. As the Hamiltonian 
is gauge invariant and linear in momentum $A_a$ and since the 
configuration observables $O$ that were extracted via the rigging map 
of the previous section are to be identified with the variables $E^a_3$,
in order to match the representation of the reduced phase space quantised
theory with the Dirac quantised theory we pick the following 
physical Hilbert space:\\
The vacuum is annihilated by the electric field $E^a\Omega=0$ and 
excited states are given by $w[F]=\exp(-i<F,A>)\Omega$ where 
\be \label{6.33}
<F,A>=\int\; d^3x\; F^a\; A_a
\ee
The Gauss constraint is solved by asking that $F$ is divergence free and 
the non-degeneracy is met by asking that $F^z$ is nowhere vanishing.
Since the Hamiltonian is linear in momentum $A_a$ its quantisation 
in this representation is straightforward by the general results of section 
\ref{s3}: Again $H$ itself is not defined but its unitary, weakly 
discontinuous one parameter group is
\be \label{6.34}
U(t)\;w[F]\Omega:=
\exp(-it H)\;\;w[F]\Omega=w[(e^{t X_H} K)(F)]\Omega
\ee
with the Hamiltonian vector field $X_H$ on the phase space with canonical pair
$(A_a,E^a)$ and $K^a(A,E)=E^a$ is the coordinate function on the phase 
space which does not depend on $A$ and thus we write $K(F)=K(.,F)=F^a$.
We have 
\be \label{6.35}
(X_H K)^a(F)=\epsilon^{abc}\partial_b f_c(F) 
\ee
which shows that the Hamiltonian flow of $H$ preserves the divergence 
free vector field densities. Let $\Lambda_1$ be the space of 1-foms and 
$V^1$ the space of divergence free 
vector (pseudo) densities then we have maps
\be \label{6.36}
f:\; V^1\to \Lambda_1;\; F\mapsto f(F),\; 
{\rm rot}:\; \Lambda_1\to V^1;\; F\mapsto f(F),\; 
\ee
where $f$ is given by (\ref{6.32a}). 
Then $X_H={\rm rot}\circ f$ and the first few terms of the Taylor 
expansion of $(e^{t X_H}\; K)$ are 
\be \label{6.37}
X_H\; K={\rm rot}[f(K)],\;       
X_H^2\; K={\rm rot}[f({\rm rot}[f(K)])],\;       
X_H^3\; K={\rm rot}[f({\rm rot}[f({\rm rot}[f(K)])])]
\ee

The relation (\ref{6.34}) states that, similar to coherent states for the 
harmonic oscillator, the quantum evolution of the states $w[F]\Omega$ 
with $F$ the eigenvalue of the electric field, stays {\it exactly on the 
classical trajectory}. The mathematical reason fo this is the same as 
for coherent states, namely that the Hamiltonian flow preserves the 
polarisation: For coherent states, holomorphic functions of 
$z=q-ip$ are preserved while here functions of $E$ are preserved.        

Furthermore, (\ref{6.34}) states that we have access to the 
{\it non-perturbative scattering matrix} in this interacting 
quantum field theory. Its matrix elements are exactly computable 
\be \label{6.37}
<w[F]\Omega,\; U(t)\; w[F']\Omega>=\delta_{F,(e^{t X_M} K)(F')}
\ee
To interpret this model physically, note that in the given gauge 
the inverse density weight two metric is explicitly given by 
\be \label{6.38}
Q^{ab}=\left( \begin{array}{ccc}
1+[E^x]^2 & E^x\; E^y & E^x\; E^z \\
E^y\;  E^x & 1+[E^y]^2 & E^y \;E^z \\
E^z\;  E^x & E^z\; E^y & [E^z]^2 \\
\end{array}
\right)
\ee
We write $E^I=h^I,\; I=1,2$ and $E^z=1+h^z$ with 
\be \label{6.39}
h^z(x,y,z)=-\int_{-\infty}^z\; ds\; (\partial_I h^I)(x,y,s)
\ee
assuming that $h^I$ vanish at infinity
so that $h^x,h^y$ are two independent polarisations of the metric
which could be called ``gravitons''. To 
first order in $h^I$ we have $Q^{ab}=\delta^{ab}
+2h^I \delta_I^{(a}\delta^{b)}_z +2 h^z \delta^a_z\delta^b_z$. 
Thus, we may interpret the matrix elements of the scattering matrix 
as graviton propagation or metric perturbation amplitudes. 
The fact that the matrix elements only take values $0,1$ between 
the Weyl states $w[F]\Omega$ looks strange at first but the fact that the 
$w[F]\Omega$ are similar to semiclassical coherent states in a Fock space 
makes this look less strange: Consider the coherent states for 
free Maxwell theory on Minkowski space given by 
\be \label{6.40}
\psi_Z=\exp(-\frac{1}{2}\int\; d^3x \; \delta_{ab} \bar{Z}^a Z^b)\;
\exp(\int\;d^3x Z^a C_a^\ast)\Omega_F,\; 
C_a=(\omega^{1/2}\; A_a-i \omega^{-1/2} E^a)/\sqrt{2},\omega^2=-\Delta 
\ee 
where $\partial_a Z^a=0$ is a complex divergence free vector field and 
$\Omega_F$ is the Fock vacuum with $C_a \Omega_F=0$ similar to our 
$E^a \Omega=0$. The normal ordered 
Maxwell Hamiltonian like ours can be expressed 
just in terms of electric and magnetic field
$H_M=\int\; d^3x \delta_{ab}\;: (E^a E^b+B^a B^b):$. Then it is well known 
and can be verified also immediately that 
$e^{-itH_M} \psi_Z=\psi_{e^{t X_{H_M} Z}},\; e^{tX_{H_M}}\; 
Z=e^{it\;\omega}\cdot Z$. Thus the modulus 
of the coherent state scattering amplitude is 
\be \label{6.41}
|<\psi_Z,U(t)\psi_{Z'}>|^2=e^{-\int\; d^3x\; 
[Z^a-e^{it\omega}\cdot Z^{a\prime}]
[Z^a-e^{it\omega}\cdot Z^{a\prime}]^\ast}
\ee
which due to the infinite number of degrees of freedom 
involved is very sharply peaked at $Z=e^{t X_{H_M}} Z'$ and thus is almost 
a Kronecker $\delta$. That it is not exactly a Kronecker $\delta$ is because 
in the Fock represention the Hamiltonian unitaries are weakly continuous.

The analogs of graviton states in our representation are therefore 
{\it not} the $w[F]\Omega$. To obtain the analog of graviton states 
we note that next to $\Omega$ there are many other ground states of the 
Hamiltonian namely $w[F]\Omega$ with $F=$const. This is because 
the Hamiltonian flow on $E^a$ is a first order equation 
$\dot{E}^a={\rm rot}[f(E)]^a$ and while $f_a(E)$ depends algebraically 
on $E$, the rotation operation maps a constant one form to zero.
Thus, given constant initial data, these stay invariant in time as well.
Thus we may pick as the Minkowski state the state 
\be \label{6.42}
\Omega_M:=w[F_M]\Omega,\; F_M^a=\delta^a_z
\ee
which indeed yields $Q^{ab}=\delta^{ab}$ (note that $F^a_M$ is divergence
free). From the point of view of 
$\Omega$, $\Omega_M$ is a highly excited and dynamically stable 
(in fact ground) state.
We consider now Fock states as excitations of $\Omega_M$. Given 
a wave vector $0\not=\vec{k}\in \mathbb{R}^3$ pick two tranversal 
real valued vectors 
$\vec{e}_\lambda(\vec{k}),\;\lambda=1,2$ such that with 
$\vec{e}_3(\vec{k})=\vec{k}/||\vec{k}||$ these three vectors constitute 
a right oriented ONB of $\mathbb{R}^3$ with Euclidian metric. 
Consider the even/odd functions under reflection at the origin 
\be \label{6.43}
F^a_{\vec{k},\lambda,+}(x)=e^a_\lambda(\vec{k})\cos(\vec{k}\cdot \vec{x}),     
F^a_{\vec{k},\lambda,-}(x)=e^a_\lambda(\vec{k})\sin(\vec{k}\cdot \vec{x})     
\ee 
and for ``occupation numbers'' $n\in \mathbb{N}$ the ``multi-particle state''
\be \label{6.44}
\Psi_{\{k,\lambda,\sigma,n\}}=
w[\sum_{l=1}^N n_l\; F_{\vec{k}_l,\lambda_l,\sigma_l}]\;\Omega_M
\ee
with pairwise different $\vec{k}_1,..,\vec{k}_N$. Then
\ba \label{6.45} 
&& <\Psi_{\{k,\lambda,\sigma,n\}},\;\Psi_{\{k',\lambda',\sigma',n'\}}>
=\delta_{
\sum_{l=1}^N n_l\; F_{\vec{k}_l,\lambda_l,\sigma_l}
,\sum_{l=1}^{N'} n'_l\; F_{\vec{k}'_l,\lambda'_l,\sigma'_l}}
\nonumber\\
&=& \delta_{N,N'}\; \sum_{\pi\in S_N}\; \prod_{l=1}^N\;
\delta_{(\vec{k}_{\pi(l)},\lambda_{\pi(l)},\sigma_{\pi(l)}),
(\vec{k}'_l,\lambda'_l,\sigma'_l)}
\ea
where $S_N$ is the symmetric group in $N$ elements. We have exploited that 
the Kronecker is non-vanishing iff the number of labels are the same 
up to a permutation of the arguments. Thus (\ref{6.45}) looks like 
a Fock space inner product except that plane wave smeared creation 
operators have been replaced by Weyl elements which are Kronecker 
normalised rather than $\delta-$distribtion normalised. If one 
replaces the plane wave smearing by wave packet smearing, the analogy 
gets even closer.

The scattering matrix elements between the ``Fock states'' (\ref{6.44}) is now 
rather non-trivial because the form of the Fock states is not at all
preserved in time and reflects both the interaction and our non-perturbative 
treatment theoreof. This may also be seen as follows. The analog of 
the Wightman N-point functions are the Heisenberg correlators 
\ba \label{6.46} 
&& W_{F_1,..,F_N}(t_1,..,t_N)
:=<\Omega,\;U(t_1)\;w[F_1]\; U(t_1)^{-1}... 
U(t_N)\;w[F_N]\; U(t_1)^{-N}\;\Omega>
\nonumber\\
&=& <\Omega,\;w[\sum_{l=1}^N\; (e^{t_k X_H}\; K)(F_l)]\;\Omega>
\nonumber\\
&=& \delta_{\sum_{l=1}^N\; (e^{t_k X_H}\; K)(F_l),0}
\ea
which cannot be written as a Gaussian in the $F_l$ and thus 
does not correspond to a (quasi-)free state, thereby demonstrating that this 
is {\it an interacting QFT in 3+1 dimensions}. Using the theorems of 
section \ref{s3} we note that we can actually quantise {\it all}
interacting QFT in {\it any} dimension as long as in the Hamiltonian formulation
it is at most linear in momentum. We will come back to this in section
\ref{s10}.
\\
The amplitudes (\ref{6.46}) motivate a path integral formulation 
that we will address in section \ref{s7}. Before we communicate 
another interesting observation.

\section{Non-relational weak quantum Dirac observables}
\label{s8}

In \cite{40} for the SU(2) theory (Euclidian GR) the observation was made 
that the seven constraints $G_j=D_a=C=0$ can be solved by making the 
Ansatz (which is w.l.g. for non-degenerate metrics) 
$B^a_j=\lambda_{jk} \; E^a_k$ and $D_a=0$ requires $\lambda_{[jk]}=0$,
$C=0$ requires $\delta^{jk}\lambda_{jk}=0$ and the Bianchi identity 
${\cal D}_a B^a_j=0$ requires 
$E^a_j {\cal D}_a \lambda_{jk}=0$ when $G_j={\cal D}_a E^a_j=0$ holds where 
${\cal D}_a$ is the covariant derivative of $A$. For the SU(2) theory
it is very 
difficult to solve the condition $E^a_j {\cal D}_a \lambda_{jk}=0$. However,
we may apply the same Ansatz in the U(1)$^3$ theory and find that 
$B^a_j=\lambda_{jk} E^a_k$ solves all constraints when $\lambda_{jk}$ is 
a constant, symmetric, trace free matrix. 

Now consider the question to construct (weak) Dirac observables $O$ that
depend only on $E$. Thus $\{D[u],O\}=\{C[N],O\}=0$ 
for all $u,N$ 
at least on the 
constraint surface $G_j=D_a=C=0$ of the $A,E$ phase space. Note that 
$O$ is trivially invariant under Gauss gauge transformations.
This leads to the 
conditions 
\be \label{8.1}
F_{ab}^j E^b_j=\epsilon_{jkl} F_{ab}^j E^a_k E^b_l=0,\; F_{ab}^j=
2\partial_{[a} \frac{\delta O}{\delta E^{b]}_j}
\ee     
The relation to the above Ansatz becomes now clear because the 
requirement
\be \label{8.2}
\epsilon^{abc} \partial_b \frac{\delta O}{\delta E^c_j}=\lambda_{jk} E^a_k
\ee
with $\lambda_{jk}$ constant, symmetric, tracefree solves (\ref{8.1}) and 
the Bianchi identity when the Gauss constraint holds. 
The question is whether (\ref{8.2}) can be solved 
for $G$. The surprising answer is that it can for any such matrix $\lambda$
which enables us construct {\it an infinite number of pure configuration,
non-relational Dirac observables}.

We claim that on the constraint surface defined by the Gauss constraint
\be \label{8.3}
O_\lambda[E]:=\frac{1}{2}\lambda^{jk}\;\int_\sigma\; d^3x\;\int\; d^3y\;
E^a_j(x)\; \kappa_{ab}(x,y)\; E^b_k(y)
\ee
solves (\ref{8.2}) with the symmetric Green function 
$\kappa_{ab}(x,y)=\kappa_{ba}(y,x)$ 
\be \label{8.4}
\kappa_{ab}(x,y):=-\epsilon_{acb} \; 
\delta^{cd}\;\frac{\partial}{\partial x^d} \;\kappa_\Delta(x,y),\;
\kappa_\Delta(x,y)=-\frac{1}{4\pi||x-y||},\; 
||x-y||^2=\delta_{ab}(x^a-y^a)(x^b-y^b)
\ee
Obviously $\kappa_\Delta$ is the Green function of the Laplace operator 
in flat Euclidian space. The astonishing fact is that $O_\lambda$ 
is (weakly) spatially diffeomorphism invariant even though 
$\kappa_{ab}$ heavily 
depends on the Euclidian background metric $\delta_{ab}$. A similar 
counter-intuitive effect is oserved for knot invariants \cite{41}
which are also constructed using the background dependent Green
function (\ref{8.4}). We have 
\ba \label{8.5}
&& \epsilon^{acb}\frac{\partial}{\partial x^c}\; \kappa_{be}(x,y)
=-[\delta^a_d\;\delta^c_e-\delta^a_e\;\delta^c_d]\; \delta^{df}\;
\frac{\partial^2}{\partial x^c \partial x^f} \; \kappa_\Delta(x,y)
\nonumber\\
&=&-[\delta^{ab}
\frac{\partial^2}{\partial x^b \partial x^e}
-\delta^a_e\;\Delta_x]\; \kappa_\Delta(x,y)
\nonumber\\
&=&[\delta^a_e\delta(x,y)-\delta^{ab} 
(\frac{\partial^2}{\partial x^b \partial x^e}\Delta^{-1})(x,y)
=P^a_{e\perp}(x,y)
\ea
thus $\kappa_{ab}$ is the Green function of the rotation operator 
$\vec{\nabla}\times$ on the space of divergence free vector field densities i.e.
on general vector field densities it yields the transversal projector. 

Indeed, since $\lambda$ is a symmetric tensor and $\kappa_{ab}(x,y)$ is a 
symmetric integral kernel we have 
\be \label{8.5}
\frac{\delta}{\delta E^a_j(x)} O_\lambda=\lambda^{jk}\;
[\kappa_{ab}\cdot \; E^b_k](x)
\ee
Thus by construction 
\be \label{8.6}
[{\rm rot} \frac{\delta O_\lambda}{\delta E^\cdot_j}]^a(x)
=\lambda^{jk}\; (P^a_{b\perp} \cdot E^b_k)(x)
=\lambda^{jk} E^a_k(x)-
[\kappa_{\Delta}\cdot \delta^{ab}\partial_b G_k](x)
\ee  
where $G_k=\partial_a E^a_k$ is the Gauss constraint. It follows 
that $G_\Lambda$ is a weak Dirac observable, more precisely it is 
a strict Dirac observable on the partial constraint surface defined 
by the Gauss constraint only. Note that (\ref{8.5}) of course trivially
obeys the Bianchi identity also away from the Gauss constraint surface.  

Since the constraint equations (\ref{8.1}) are linear in $O$ we can construct
an infinite number of weak Dirac observables as follows: Let
$F:\mathbb{R}^N\to \mathbb{R}$ be any $C^1(\mathbb{R}^N)$ function,
let $\lambda_I,\; I=1,..,N$ be any symmetric, tracefree, constant 
matrices then 
\be \label{8.7}
F[E]:=F(O_{\lambda_1}[E],..,O_{\lambda_N}[E])
\ee
is a strong Dirac observable on the Gauss constraint surface. Note 
that $F$ need not be a polynomial. What is really astonishing is 
that it is generally believed that a spatially diffeomorphism invariant 
and Gauss invariant function of $E^a_j$ should be built from 
curvature invariants of the metric 
$q_{ab}$ where $Q^{ab}=\delta^{jk} E^a_j E^b_k=:\det(q)\;q^{ab}$, 
e.g. as a function of the infinite tower of integrals of the form 
\be \label{8.8}
O_N:=\int\; d^3x\; \sqrt{\det(q)}\;{\rm Tr}(R(q)^N),\; 
R_{ab}\;^{cd}=R_{abef}\; q^{ec} q^{fd}
\ee
which are highly non-polynomial and rather complicated functions of 
$E^a_j$.

However, (\ref{8.7}) is not of this form and even more, it is also 
invariant under the Hamiltonian constraint. An obvious reason 
for this should be that (\ref{8.8}) is exactly invariant under 
spatial diffeomorphisms while (\ref{8.7}) is only when the Gauss 
constaint holds. Of course the space of functions (\ref{8.7}) 
is much smaller than the space of functions of the (\ref{8.8}) as 
the space of symmetric and trace-free matrices is just five dimensional
while (\ref{8.8}) involves infinitely many algebraically independent 
elements (although there are Gauss-Bonnet type relations among them,
i.e. certain linear combiantions yield topological invariants). 

On the Hilbert space spanned by the $w[F]\Omega$ constructed in 
section \ref{s3} the $O_\lambda$ are diagonal with eigenvalue 
$-O_\lambda[F]$. It follows 
\ba \label{8.9}
&&[U[u]^{-1}\; O_\lambda U[u]-O_\lambda]\;w[F]\;\Omega
=-[O_\lambda[(e^{X_u} K)(F)]-O_\lambda[F]]\;w[F]\;\Omega
=-([e^{X_u}-1]\; O_\lambda\circ K)[F]\;w[F]\;\Omega=0
\nonumber\\
&& [U[M]^{-1}\; O_\lambda U[u]-O_\lambda]\;w[F]\;\Omega
=-[O_\lambda[(e^{X_M} K)(F)]-O_\lambda[F]]\;w[F]\;\Omega
=-([e^{X_M}-1]\; O_\lambda\circ K)[F]\;w[F]\;\Omega=0
\nonumber\\
&&
\ea
because we consider divergence free $F$. More in detail, by construction
$X_u O_\lambda, \; X_M O_\lambda$ are both proportional to $G_j$ 
and $X_u G_j\propto G_j, X_M G_j=0$. Thus 
$[e^{X_u}-1]O_\lambda,\; [e^{X_M}-1]O_\lambda \propto G_j$ 
which is then evaluated at $F$ but $G_j[F]\equiv 0$. Thus 
on the span of the solutions of the Gauss constraint, the operator 
$O_\lambda$ and more generally (\ref{8.7}) stronly commutes with the 
constraints and thus preserves the physical Hilbert space constructed in 
section \ref{s5}.

\section{Path integral formulation of the reduced theory}
\label{s7}

We are interested in writing the summation (rather than integral)
kernel of the propagator $U(t)$ as a ``path integral''. Usually at 
this point one performs a Wick rotation in time $t$ to obtain a 
contraction operator $e^{-t H}$ rather than a unitary one $e^{-it H}$
which has better chances to result in a rigorously defined measure on 
the space of field histories. However this is only true if $H$ is bounded 
from below which for our $H$ is far from being the case. Thus Wick 
rotation appears of little use and we thus stick to physical rather than 
imaginary time in what follows. For this reason, our condiderations will 
be largely heuristic.\\
\\
Let ${\cal F}$ be the set of all divergence free vector field densities.
When we write $\sum_F$ we mean $\sum_{F\in {\cal F}}$ in what follows.
Then the $w[F]\Omega,\;F\in {\cal F}$ form an orthonormal basis of 
the non-separable Hilbert space $\cal H$ and we may, given $t$ invoke 
a partition of $[0,t]$ into segments $[kt/N,(k+1)t/N],\;k=0,..,N-1$ 
of length $t/N$ and $N-1$ resolutions of unity 
\ba \label{7.1}
&&<w[F_N]\Omega,\; U(t)\; w[F_0]\Omega>     
=\sum_{F_1,..,F_{N-1}}\; 
\prod_{k=0}^{N-1} \;<w[F_{k+1}]\Omega,\; U(\frac{t}{N})\; w[F_k]\Omega>
\nonumber\\
&=& 
\sum_{F_1,..,F_{N-1}}\; 
\prod_{k=0}^{N-1} \;
\delta_{F_{k+1},(e^{\frac{t}{N} X_H}\; K)(F_k)}
\ea
Note that $\delta_{F,F'}$ really means 
\be \label{7.2}
\delta_{F,F'}=\prod_{x\in\mathbb{R}^3,a=1,2,3}\; 
\delta_{F^a(x),F^{a\prime}(x)}
\ee
where the product over all $a=1,2,3$ is redundant because $F,F'$ are 
divergence free but since these are Kronceker functions rather than 
$\delta$ distributions, the redundant 
factors of unity do not cause any singularity. 

The Bohr measure on the Bohr compactification $\mathbb{R}_B$ 
of the real line 
has the property 
\be \label{7.3}
\mu_B(T_k)=\delta_{k,0},\; T_k:\; \mathbb{R}_B\mapsto \mathbb{C},\;
y\mapsto e^{iky}
\ee
Consider a partition $P_\epsilon$ 
of $\mathbb{R}^3$ into cells $\Box$ of coordinate 
vcolume $\epsilon^3$ and centre $p_\Box$ 
and denote by $\lim_\epsilon$ the limit as the partition 
reaches the continuum. Then 
\ba \label{7.4}
&&\delta_{F,F'}
=\lim_\epsilon \; \prod_{\Box\in P_\epsilon,a}\;
\delta_{\epsilon^3 F^a(p_\Box),\epsilon^3 F^{a\prime}(p_\Box)}
=\lim_\epsilon \; \prod_{\Box\in P_\epsilon,a}\;
\mu_B(T_{\epsilon^3[F^a(p_\Box)-F^{a\prime}(p_\Box)]})
\nonumber\\
&=& \lim_\epsilon \; \mu_B(\prod_{\Box\in P_\epsilon,a}\;
T_{\epsilon^3[F^a(p_\Box)-F^{a\prime}(p_\Box)]})
\nonumber\\
&=:&\mu_B(\exp(<.,F-F'>)=:\int_{{\cal C}} d\mu_B(C) \; 
e^{i\int \; d^3x\; C_a(x)[F^a(x)-F^{a\prime}(x)]}
\ea
where we introduced the product Bohr measure on products of the Bohr line
and $\cal C$ denotes the space of ``Bohr connections''. Thus (\ref{7.1})
becomes
\be \label{7.5}
<w[F_N]\;\Omega,\;U(t)\;w[F_0]\;\Omega>
=\sum_{F_1,..,F_{N-1}}\; \int\; \prod_{l=0}^{N-1} d\mu_B(C^l)\;
\exp(i\sum_{l=0}^{N-1}
\int\;d^3x\;C^l_a(x)[F^a_{l+1}-(e^{\frac{t}{N}H}\; K)^a(F_l)](x))
\ee
Setting $F^a(\frac{lt}{N},x):=F^a_l(x),\;
\frac{t}{N}\;\dot{F}^a_l(\frac{tl}{N},x):=F^a_{l+1}(x)-F^a_l(x),\;
C_a(\frac{lt}{N},x):=C_a^l(x)$
we have due to linearity of $X_H$ in $A$
\be \label{7.6}
\int\;d^3x\;C^l_a(x)[F^a_{l+1}-(e^{\frac{t}{N}X_H}\; K)^a(F_l)](x))
=\frac{t}{N}   
\int\;d^3x\;[C_a \dot{F}^a-H(C,F)](\frac{tl}{N},x)+O((t/N)^2
\ee
To free ourselves from the restriction $\partial_a F^a=0$ when summing over 
$F$ we introduce a Kronecker $\prod_l \delta_{G_l}$ 
and remove by Bohr integration with respect to a zero component $C_0$ 
of the connection. Using these manipulations we find
\be \label{7.7}
<w[F_N]\;\Omega,\;U(t)\;w[F_0]\;\Omega>
=\sum_{F_1,..,F_{N-1}}\; \int\; \prod_{l=0}^{N-1} d\mu_B(C^l)\;
\exp(i\frac{t}{N}\;\sum_{l=0}^{N-1}
\int\;d^3x\;[C_a\dot{F}^a-C_0\; \partial_a F^a -H(C,F)](\frac{lt}{N},x)
\ee
Formally taking the limit $N\to \infty$ keeping final $F_N=F_f=F(t)$ and 
initial $F_0=F_i=F(0)$ fixed and integrating by parts we find 
\be \label{7.8}
<w[F_f]\;\Omega,\;U(t)\;w[F_i]\;\Omega>
=\int\; [d\mu_D(F)]\;\int [d\mu_B(C)]
\exp(-i\int_0^t\; d^3x\; \{[\partial_0 C_a-\partial_a C_0] F^a
+H[C,F]\}
\ee
where $\int\; [d\mu_D(F)]:=\sum_{[F]}$ is the discrete (counting) measure.
The exponent depends on $F^a$ which interpret as the electric field 
$E^a$, the magnetic field or spatial spatial curvature 
$B^a=\epsilon^{abc}\partial_b C_c$ of $C$ and the 
temporal spatial curvature $\partial_0 C_a-\partial_a C_0$. 

Performing the integral over $C_\mu$ 
yields back $\partial_a F^a=0$ and the classical equations of motion 
$\dot{F}^a=\delta H/\delta C_a$ while performing the integral over $E$,
if it could be done explicitly, would yield a pure connection formulation
\cite{27}.

\section{Spin foam model}
\label{s9}
 
As is well known \cite{37,42} the heuristic relation between the rigging map
$\eta$ and the path integral over the full unconsrained phase space 
is as follows:\\
Let $q^a,p_a$ be canonical coordinates on the full phase space, let 
$F_\alpha$ be all first class constraints, let $S_I$ be all second class 
constraints and let $G^\alpha$ be a complete system of gauge fixing 
conditions for the $F_\alpha$. Split 
$(q^a)=(\phi^\alpha,x^I,Q^A),\;
(p_a)=(\pi_\alpha,y_I,P_A)$ where $Q^A,P_A$ are referred to as the true 
degrees of freedom. Then 
\be \label{9.1}
<\eta \psi,\; \eta\psi'>:=\lim_{T\to \infty}\;
\frac{
\int\;[dq\;dp]\;\psi^\ast(\phi_T,Q_T)\;
\psi'(\phi_{-T},Q_{-T})\;\delta[F]\;\delta[S]\;\delta[G]\;
|\det[\{F,G\}]|\;|\det[\{S,S\}]|^{1/2}\; e^{i\int_{-T}^T\;dt p_a\dot{q}^a}
}
{\psi,\psi'\to \Omega_0}
\ee
where $\Omega_0$ is some cyclic reference vector. The notation is 
that square brackets denote the product over all $t\in [-T,T]$ of 
instantaneous quantities which are denoted by round brackets, 
e.g. $\delta[F]=\prod_t \delta(F_t)$ and  
$\delta(F_t)=\prod_x \delta(F_t(x))$.  

In \cite{27} the first and second class classification of all constraints 
following from the covariant action (\ref{2.1}) has been derived so that 
we can identify the above structures for the U(1)$^3$ model.
\begin{itemize}
\item[1.] The unconstrained phase space has configuration variables 
$(q^a)=(A_\mu^j,\; \hat{e}^\mu_j,\;\hat{e}^\mu_0)$ and momentum variables
$(p_a)=(M^\mu_j,\; \hat{I}_\mu^j,\;\hat{I}_\mu^0)$ with spacetime 
indices $\mu,\nu,\rho,..=t,1,2,3$ and $j=1,2,3$. 
The relation between $\hat{e}^\mu$ and 
$e^\mu$ is that $\hat{e}^\mu_\cdot=|\det(e_\mu^\cdot)|^{1/2} \; e^\mu_\cdot$
is a half density.
\item[2.] Let
\be \label{9.2}
\sigma^{\mu\nu}_j:=2\;\hat{e}^{[\mu}_0\; \hat{e}^{\nu]}_j+
\epsilon^{jkl}\;\hat{e}^{[\mu}_k\; \hat{e}^{\nu]}_l
\ee
Primary constraints are  
\be \label{9.2}
T^a_j=M^a_j-\sigma^{ta}_j,\;M^t_j,\; \hat{I}_\mu^j,\;\hat{I}_\mu^0
\ee
with $a,b,c,..=1,2,3$ spatial indices. 
\item[3.] The first nine primary constraints 
are the triad conditions. Then $M^t_j$ and 
seven of the sixteen $\hat{I}_\mu^j,\hat{I}_\mu^0$ are stabilised 
by the secondary constraints
\be \label{9.3}
G_j=\partial_a \sigma^{ta}_j,\;
D_a=F_{ab}^j \hat{e}^b_j,\; C=\epsilon_{jkl}F_{ab}^j \hat{e}^b_k\;
\hat{e}^c_k,\;T_j=\hat{e}^t_j
\ee
with $F_{ab}^j=2\partial_{[a} A_{b]}^j$. Note that $T_j=0$ is the 
``time gauge'' 
which in the U(1)$^3$ theory is a necessary constraint and not a 
convenient gauge fixing condition.  
\item[4.] The Lagrange multiplierst $v_a^j, v_t^j,\hat{v}^\mu_j,\hat{v}^\mu_0$
of the primary constraints get fixed to
\be \label{9.4}
v_a^j=[\hat{e}^t_0]^{-1}\;[F_{ab}^j\hat{e}^b_0+\epsilon_{jkl} F_{ab}^k 
\hat{e}^b_l],\;\hat{v}^t_j=0, 
\hat{v}^a_j=-[\hat{e}^t_0]^{-1}[\partial_b\sigma^{ab}_j+\hat{v}^t_0\hat{e}^a_j]
\ee
to stabilise the remaining nine of the $\hat{I}_\mu^j,\hat{I}_\mu^0$,
the time gauge $T_j$ and the triad conditions $T^a_j$,
while $v_t^j,\hat{v}^\mu_0$ remain free. No tertiary constraints arise.
\item[5.] The 14 first class constraints are given by 
\ba \label{9.5}
&& \hat{I}_\mu^0,\;M^t_j,\; \hat{G}_j=\partial_a M^a_j,\; 
\nonumber\\
\hat{D}_a &=& F_{ab}^j-A_a^j\hat{G}_j+
\sum_{L=0}^3\;\{2\;[
(\hat{I}_a^L \hat{e}^b_L)_{,b}-\hat{I}_{b,a}^L \hat{e}^b_L)_{,b}-  
(\hat{I}_b^L \hat{e}^b_L)_{,a}]+\frac{1}{2}
(\hat{I}_t^L \hat{e}^t_{L,a}-\hat{I}_{t,a}^L \hat{e}^t_L)\}
\nonumber\\
\hat{C} &=& v_t^j\;M^t_j+v_a^j T^a_j+\sum_L \hat{v}^\mu_I \hat{I}_\mu^L
-A_t^j \partial_a\sigma^{ta}_j-\frac{1}{2} F_{ab}^j \sigma^{ab}_j
\ea
where in the last equation the above fixed values for 21 of 28 Lagrange 
multipliers have to be assumed. The 12 second classs pairs are 
\be \label{9.6}
\{\hat{I}_t^j(x),T_k(y)\}=\delta^j_k \delta(x,y),\;
\{T^a_j(x),\hat{I}_b^k(y)\}=\delta^a_b\delta_j^k\hat{e}_0^t(x)\delta(x,y)
\ee
\end{itemize}
The unconstrained phase space thus has 2(12+16)=56 degrees of freedom. The 
14 first class constraints and 24 second class constraints remove 
14+14+24=52 degrees of freedom leaving 4 physical degrees of freedom, 
i.e. 2 canonical pairs. We thus identify the canonical structure as follows:
\begin{itemize}
\item[1.] First class set
\be \label{9.7}
\{F_\alpha\}=\{\hat{I}_\mu^0,\; M^t_j,\;\hat{G}_j, \; \hat{D}_a,\;\hat{C}\}
\ee
where the latter 7 constraints are solved for $M^{a\parallel}_j,A_{a\perp}^j$.
\item[2.] Corresponding gauge fixing conditions 
\be \label{9.8}
\{G^\alpha\}=\{\hat{e}^\mu_0,\;A_t^j,\; A^j_{a\parallel},\;M^{a\perp}_{1,2}\}
\ee
\item[3.] Second class set
\be \label{9.9}
\{S_I\}=\{(T_j,\hat{I}_t^j),(\hat{I}_a^j,\;T^a_j)\}
\ee
which are solved for $\hat{e}^t_j,\;\hat{I}^j_t,\;\hat{I}_a^j,\; \hat{e}^a_j$.
\item[4.] The split of canonical pairs is thus
\ba \label{9.10}
\{(\phi^\alpha,\pi_\alpha)\} &=& 
\{(\hat{e}^\mu_0,\hat{I}_\mu^0),\;(A_t^j,M^t_j),\;
(A_{a\parallel}^j,M^{a\parallel}_j),\;(A_{a\perp}^{1,2}, M^{a\perp}_{1,2})\}
\nonumber\\
\{(x^I,y_I)\} &=& 
\{(\hat{e}^t_j,\hat{I}^j_t),\;(\hat{e}^a_j,\hat{I}_a^J)\}
\nonumber\\
\{Q^A,P_A\} &=& \{(A_{a\perp}^3,M^{a\perp}_3)\}
\ea
\end{itemize}
The aim is now to rewrite (\ref{9.1}) in a covariant form, i.e. we 
aim at keeping only the variables $A_\mu^j,\hat{e}^\mu_j,\hat{e}^\mu_0$ 
that appear in the classical action (\ref{2.1}). Thus we want 
to get rid of $M^\mu_j,\hat{I}_\mu^j,\hat{I}_\mu^0$.

First we note that 
\be \label{9.11}
\det[\{S,S\}]=[(\hat{e}^t_0)^9],\;
\det[\{F,G\}]=\det[\{F_s,G_s\}],\; \det[\{F_p,G_p\}]=1
\ee
where $F_s,G_s$ denote the list of secondary first class constraints and 
gauge fixing conditions only, i.e. 
$\{F_s\}=\{\hat{G}_j,\hat{D}_a,\hat{C}\}$ and 
$\{G_s\}=\{A_{a\parallel}^j,M^{a\perp}_{1,2}\}$ while
$F_p,G_p$ denote the list of primary first class constraints and 
gauge fixing conditions only, i.e. 
$\{F_p\}=\{\hat{I}_\mu^0,M^t_j\}$ and 
$\{G_p\}=\{\hat{e}^\mu_0,A_t^j\}$ 

The Liouville measure $[dF_p dG_P]$ is invariant under canonical 
transformations (at each time) and the symplectic potential 
$F_P \;dG_p$ changes by a an exact differential, hence its time integral
is invariant for canonical transformations that decay at vearly early and 
late times. Therefore, by the Fadeev-Popov method we can drop 
$\delta[G_p]\;|[\det\{F_p,G_p\}]|$ from the numerator and denominator path 
integral and we can then carry out the integral $\int [dF_p]\; \delta[F_P]=1$
Thus altogether we can drop $[dF_p] \; \delta[G_p]\;\delta[F_p]\;
|\det[\{F_p,G_p\}|$. 

Next, $Q_{\pm T}$ is invariant under asymptotically trivial gauge 
transformations generated 
by $F_s$ so we can identify $Q$ with the $F_s$ gauge invariant projection
or relational obeservable $O_Q:=O_Q^{G_s,F_s}$ corresponding to the 
gauge fixing $G_s$. Therefore  
\be \label{9.12}
B_T(O_Q):=\psi^\ast((O_Q)_T)\psi'((O_Q)_{-T}),\;
B^0_T(Q):=\Omega_0^\ast((O_Q)_T)\Omega_0((O_Q)_{-T})
\ee
are gauge invariant. By the same Fadeev-Popov argument as above, 
we can thus drop $\delta[G_s]\;|\det[\{F_s,G_s\}|$ from numerator and 
denominator path integral.

Furthermore, we can carry out the interals over 
$\hat{I}_\mu^j,\hat{I}_\mu^0,M^\mu_j$ enforcing $\hat{I}_\mu^j=0,\;
M^t_j=0,\; M^a_j=\sigma^{ta}_j$. This simplifies (\ref{9.1})
to  
\be \label{9.13}
\lim_{T\to \infty}\;
\frac{
\int\;[dA_\mu^j\;d\hat{e}^\mu_j\;d\hat{e}^\mu_0]\;
\delta[(F_s)_{T^a_j=0}]\;\;\delta[T_j]\; [|\hat{e}^t_0|^{9/2}|]
\exp(i\int_{-T}^T\; dt\; \sigma^{ta}_j\; dA_a^j)\; B_T(Q)
}
{ B_T(Q)\to B_T^0(Q)}
\ee
We note that the secondary first class constraints at $T^a_j=\hat{P}_\mu^L
=M^t_j=0$ simplify to 
\be \label{9.14}
\hat{G}_j=\partial_a\sigma^{ta}_j,\;
\hat{D}_a=F_{ab}^j \sigma^{ta}_j-A_a^j\hat{G}_j,\;
\hat{C}=-A_t^j\; \hat{G}_j-\frac{1}{2} F_{ab}^j\;\sigma^{ab}_j
\ee
We bring the constraints $F_s$ into the exponents using corresponding 
Lagarange multipliers $\lambda^j,\lambda^a,\lambda$ 
which turns (\ref{9.1}) into the form
\be \label{9.15}
\lim_{T\to \infty}\;
\frac{
\int\;[dA_\mu^j\;d\hat{e}^\mu_J\;d^7\lambda]\;
%d\lambda^j\;d\lambda^a\;d\lambda]\;
%d\hat{e}^\mu_0]\;
[|\hat{e}^t_0|^{9/2}]\; \delta[\hat{e}^t_j]
\exp(i\int_{|t|\le T}\;d^4x\;  
[\sigma^{ta}_j\; \dot{A}_a^j+[\lambda^j-\lambda^a A_a^j+A_t^j]\hat{G}_j+
[\lambda^a\sigma^{tb}_j+\frac{\lambda}{2} \sigma^{ab}_j] F_{ab}^j)\; B_T(Q)
}
{ B_T(Q)\to B_T^0(Q)}
\ee
We shift $A_t^j+\lambda^j-\lambda^a A_a^j\to \tilde{A}_t^j$ and drop the 
tilde again. Then all dependence of the integrand on $\lambda^j$
has disappeared, yielding an infinite constant which drops out of the fraction
leaving us with 
\be \label{9.16}
\lim_{T\to \infty}\;
\frac{
\int\;[dA_\mu^j\;d\hat{e}^\mu_J\;d\lambda^a\;d\lambda]\;
%d\hat{e}^\mu_0]\;
[|\hat{e}^t_0|^{9/2}]\; \delta[\hat{e}^t_j]
\exp(i\int_{-T}^T\; dt\;\int\; d^3x\;  
[\sigma^{ta}_j\; F_{ta}^j+
[\lambda^a\sigma^{tb}_j+\frac{\lambda}{2} \sigma^{ab}_j] F_{ab}^j)\; B_T(Q)
}
{ B_T(Q)\to B_T^0(Q)}
\ee
where we integrated by pars in the exponent.
Since the support of the integrand is at $\hat{e}^t_j=0$ this
reduces $\sigma^{\mu\nu}_j$ to
\be \label{9.17} 
\sigma^{ta}_j
%=2\hat{e}^{[t}_0 \hat{e}^{a]}_j+\epsilon_{jkl} \hat{e}^t_k\; \hat{e}^a_l
=\hat{e}^t_0\hat{e}^a_j,\;
\sigma^{ab}_j=2\hat{e}^{[a}_0 \hat{e}^{b]}_j+\epsilon_{jkl} 
\hat{e}^a_k\; \hat{e}^b_l
\ee
We perform a field redefinition as follows
\ba \label{9.18}
&& \lambda=\lambda',\; \lambda^a=\lambda^{a\prime}
\nonumber\\ 
&& \hat{e}^t_j=\hat{f}^t_j,\;
\hat{e}^a_j=\hat{f}^a_j/\sqrt{\lambda'},\; 
\hat{e}^t_0=\sqrt{\lambda'}\; f^0_t,\;
\hat{e}^a_0=
[\hat{f}^a_0-\sqrt{\lambda'}\; \lambda^{a\prime}\;\sqrt{f}^0_t]/\lambda'
\ea
whose pointwise Jacobean is easily computed to be $(\lambda')^{-7}$. Note
also $[|\hat{e}^t_0|^{9/2}]=[|\hat{f}^t_0|^{9/2}]\;[|\lambda'|^{9/2}]$.
After this transformation, as one can check,
the exponent no longer depends on $\lambda',
\lambda^{a\prime}$, all dependence on those
is a fixed function multiplying the integrand which can thus be integrated 
out and yield cancelling 
factors. Relabelling $\hat{f}$ by $\hat{e}$ again we find that (\ref{9.1})
becomes
\be \label{9.19}
\lim_{T\to \infty}\;
\frac{
\int\;[dA_\mu^j\;d\hat{e}^a_j]\;d\hat{e}^\mu_0]\;
[|\hat{e}^t_0|^{9/2}|]\; \delta[\hat{e}^t_j]
\exp(\frac{i}{2}\int_{-T}^T\; dt\;\int\; d^3x\;  
F_{\mu\nu}^{KL} \; \hat{e}^\mu_K\;\hat{e}^\nu_L)
\; B_T(Q)
}
{ B_T(Q)\to B_T^0(Q)}
\ee
which is almost what one would expect (i.e. the integrand is the 
exponent of the classical action) except for the measure factor 
$[|\hat{e}^t_0|^{9/2}]$ and the fact that the integral is supported 
on the time gauge $\hat{e}^t_j=0$.\\ 
\\
Expression (\ref{9.19}) has an exponent linear in the connection 
and thus integrating over it yields an integral over $\hat{e}^\mu_L$ 
supported at the solution of its classical field equation. However,
to test techniques of spin foam models \cite{43} in this Abelian U(1)$^3$ 
gauge group theory one rather aims at a BF formulation of (\ref{9.19})
rather than a tetrad formulation. This will bring out the significance 
of the time gauge: The imposition of the tetrad time gauge translates
into a {\it covariant simplicity constraint} on the $B$ field.
To see this we note that 
\be \label{9.20}
F_{\mu\nu}^{KL} \hat{e}^\mu_K\;\hat{e}^\nu_L
=F_{\mu\nu}^j\;[2\hat{e}^\mu_0\;\hat{e}^\nu_j
+\epsilon_{jkl}
\hat{e}^\mu_k\;\hat{e}^\nu_l]
=:F_{\mu\nu}^j\;\sigma^{\mu\nu}_j
\ee
i.e. the $B$ field of the BF formulation of the U(1)$^3$ model is 
constrained to be of the form
\be \label{9.21}
B^{\mu\nu}_j=\sigma^{\mu\nu}_j
\ee
To avoid confusion, note that this $B$ field has nothing to do with the 
magnetic field of $A$.
The left hand side has 18 degrees of freedom, the right hand side only 
13 with the time gauge imposed, thus there must be 5 simplicity 
constraints on $B$ that ensure that it has the form (\ref{9.21}). To discover
those, we try to solve (\ref{9.21}) for $\hat{e}^\mu_L$ while 
$\hat{e}^t_j=0$. We have 
\be \label{9.22}
B^a_j:=B^{ta}_j=\hat{e}^t_0\; \hat{e}^a_j,\;
B^{ab}_j=2\hat{e}^a_{[0} \hat{e}^b_{j]}+\epsilon_{jkl} \hat{e}^a_k\hat{e}^b_l
\ee
Eliminating $\hat{e}^a_j$ from the second equation via the first,
dualising $B^{ab}_j:=\epsilon^{abc}\;\tilde{B}_c^j$ and using the 
formal inverse $B_a^j B^a_k=\delta^j_k$ the second equation in 
(\ref{9.22}) can be written as 
\be \label{9.23}
(\hat{e}^t_0)^2\; \tilde{B}_a^j=\hat{e}^t_0\; \epsilon_{abc} \hat{e}^b_0
B^c_j+\det(B)\;B_a^j
\ee
and contracting with $B^a_k$ finally yields
\be \label{9.23}
(\hat{e}^t_0)^2\; \tilde{B}_a^j\; B^a_k=
\det(B)\;[\delta_{jk}+\hat{e}^t_0\; \epsilon_{jkl} \hat{e}^b_0\; B_b^l]
\ee
These are 9 equations for 4 unknowns $\hat{e}^\mu_0$ and they can be solved 
if and only if the matrix $M_{jk}:=\tilde{B}_a^j B^a_k$ has no symmetric trace 
free piece (besides $B^a_j$ being non-degenerate, which is classically
granted if $\hat{e}^\mu_I$ is a half densitised tetrad given the first 
relation in (\ref{9.22})). If that is the case we find 
\be \label{9.24}
\hat{e}^t_0=\frac{3\det(B)}{\delta^{jk} M_{jk}},\;
\hat{e}^a_0=\frac{\hat{e}^t_0\;\epsilon^{jkl} M_{kl}}{\det(B)} B^a_j
\ee
Surprisingly the condition that the tracefree symmetric part of 
$M_{jk}$ vanishes can be stated {\it covariantly}. To see this we 
compute the covariant density
\be \label{9.25}
S_{jk}(B):=\epsilon_{\mu\nu\rho\sigma} \; B^{\mu\nu}_j \;B^{\rho\sigma}_k
=4\epsilon_{abc}\; B^{ta}_{(j} B^{bc}_{k)}     
=8\epsilon_{abc}\; M_{(jk)}
\ee
which directly yields the symmetric piece of $M$. Thus the five  
U(1)$^3$ {\it simplicity constraints} read
\be \label{9.26}
S(B)-\frac{1_3}{3}\; {\rm Tr}(S(B))=0     
\ee
Conversely, inserting (\ref{9.21}) with time gauge installed into 
$S(B)$ yields that 
\be \label{9.27}
S(B)=6\; \hat{e}^t_0\; \det(\{\hat{e}^a_j\}) \; 1_3
\ee
has only a trace part. 

To rewrite the rigging map or path integral (\ref{9.19}) from 
the tetrad into the BF formulation we thus proceed as follows:
We define the map from the 18 dimensional space of tensors
$B^{\mu\nu}_j$ 
to the 5 dimensional space of trace free symmetric matrices
\be \label{9.28}
B\mapsto S_T(B):=S(B)-1_3\;{\rm Tr}(S(B))/3
\ee
and we use the first equation in (\ref{9.22}) and (\ref{9.24}) to 
define a map from tensors $B^{\mu\nu}_j$ to a 13 dimensional space
of tetrads in time gauge 
\be \label{9.29}
B\mapsto \hat{E}(B); \hat{E}^t_j\equiv 0  
\ee
Note that (\ref{9.29}) just uses the trace and antisymmetric part of 
$B^a_j \tilde{B}_{ak}$ while $S_T(B)$ is the trace free symmetric part.

Therefore (\ref{9.28}), (\ref{9.29}) defines an invertible map 
$B\mapsto (S_T(B), \hat{E}(B))$ whose inverse we write as 
$(M,\hat{e})\mapsto b(M,\hat{e})$. Then we determine a measure 
factor $\rho(B)$ such that 
\ba \label{9.30}
&&\int\; d^{18}B\; \rho(B)\;\delta^{(5)}(S_T(B)) \; f(B)
=
\int\; d^{13}\hat{e}\;d^5M \;
|\det(\partial(b(M,\hat{e})/\partial(M,\hat{e})|\; \delta^5(M)\;
\rho(b(M,\hat{e}))\; f(b(M,\hat{e}))
\nonumber\\
&=& 
\int\; d^{13}\hat{e}\;
|\det(\partial(b(M,\hat{e})/\partial(M,\hat{e})|_{M=0}\;
\rho(b(0,\hat{e}))\; f(b(0,\hat{e}))
\nonumber\\
&=& 
\int\; d^{13}\hat{e}\;|\hat{e}^t_0|^{9/2}\; f(b(0,\hat{e})
\ea
where the identity $S_T(b(M,\hat{e})=M$ was used, thus 
\be \label{9.31}
\rho(b(0,\hat{e}))=
\frac{|\hat{e}^t_0|^{9/2}}{
|\det(\partial(b(M,\hat{e})/\partial(M,\hat{e})|_{M=0}}=:J(\hat{e})
\ee
Thus for instance
\be \label{9.32}
\rho(B):=J(\hat{E}(B))
\ee
where the identity $\hat{E}(b(M,\hat{e}))=\hat{e}$ was used. The details
are reserved for future investigations. The investigations of section 
\ref{s7} suggest that the naive Lebesgue measures $dA, d\hat{e}, dB$ 
should rather be replaced by the Bohr $d\mu_B(A)$ and discrete 
$d\mu_D(\hat{e})\; d\mu_D(B)$ meeasures respectively.\\
\\
Thus we have brought the U(1)$^3$ model into the starting point for the 
usual spin foam treatment using the simplicity constraint (\ref{9.26}).
It is possible that the simpler Abelian context and the fact that the 
canonical treatment could be carried out in all details can help 
to deepen the relation between canonical and covariant LQG. 

\section{Summary and Outlook}
\label{s10}

That quantum U(1)$^3$ theory can be developed to such an extent as 
has been layed out in this paper may seem astonishing at first. In 
retrospect, however, there is a single and simple mathematical reason 
for this: The constraints or physical Hamiltonian of the theory are 
at most linear in one of the variables of a canonical pair. This
makes the system behave almost classically in the polarisation 
(in the sense of geometric quantisation \cite{44}) adapted to this partly 
linear structure. 

More realistic theories such as GR in either signature
are typically at least quadratic in both variables of a canonical 
pair and therefore this remarkable simplicity of the quantum solution of 
U(1)$^3$ theory cannot be expected fo such theories. Yet, the following 
lessons most likely can be transferred to those more ralistic theories:
\begin{itemize}
\item[1.] {\it Exponentiated vs. infinitesimal action}\\
It is the action of the {\it exponentiated} Hamiltonian constraint, and not 
its generator, that can be implemented without anomalies 
on the kinematical Hilbert space.
In LQG, that was already known for the much simpler spatial diffeomorphism 
constraint.
For the Hamiltonian constraint the situation was less clear because while 
for the spatial diffeomorphism constraint there is one ``missing smearing 
dimension'', for the Hamiltonian constraint the number of smearing 
dimensions is precisely correct in its density unity 
version. The more fundamental reason for 
failure of existence of the generator of the anomaly free action of the 
Hamiltonian constraint is weak discontinuity. 
\item[2.] {\it ``Off-shell'' kinematical vs. dual action}\\
The action of the Hamiltonian constraint can be defined directly 
on the kinematical Hilbert space, it does not require any dual 
states (distributions) or non standard operator toplogy relying on 
spatial diffeomorphism invariance. The 
algebroid closes off-shell in that sense. 
\item[3.] {\it Quantum non-degeneracy}\\
The natural and, for non-integer density weight and/or density weight less 
than two necessary, common, dense 
and invariant domain on which the quantum hypersurface deformation algebroid 
acts consists of states which are quantum non-degenerate, i.e. the volume 
operator of any open region has non-zero volume expectation 
(even eigen-)value. This is not only semiclassically expected as the 
very definition of the hypersurface deformation algebroid and its 
derivation (the quite non-trivial Poisson bracket calculation makes 
crucial use of classical non-degeneracy at every single step of the 
computation) critically rely on npn-degeneracy, in retrospect, this 
also explains why in LQG an anomaly free quantum constraint algebroid 
is so difficult to obtain: The spin network states (SNWS) which one uses as 
common, dense and invariant domain fail to be quantum 
non-degenerate. Therefore the quantum analogs of substantial 
parts of the classical Poisson bracket
calculation in terms of commutators which depend on integrals (understood 
as infinite Riemann sums) and
integrations by parts are impossible to reproduce on SNWS which can at 
most depict finitely many terms of the actually necessary infinite 
number of terms in the Riemann  
sums. The exploration of a possible quantum non-degenerate domain within 
the standard LQG setiing is currently under study \cite{30}. 
\item[4.] {\it Density weight}\\
In the U(1)$^3$ model one can work with any density weight on the 
non-degenerate domain, thereby showing that the natural density weight unity 
is not ruled out by the quantum theory. In more realistic theories 
with at least quadratic dependence in all canonical variables, density weight
unity is in fact the only choice \cite{30}. 
\item[5.] {\it Detailed action of the Hamiltonian constraint}\\
The Hilbert space representation of LQG and the one 
used here for the U(1)$^3$ are in some sense very similar as they are 
both of the discontinuous Narnhofer-Thirring type. In both representations
one can solve the SU(2) or U(1)$^3$ Gauss constraint by considering 
appropriate subspaces defined by a restriction on the smearing functions 
of the connection (``form factors'').

However, in LQG the connection is smeared along 1d paths while here we use 
a 3d smearing. This is possible because the constraints are linear in the 
connection, the quadratic dependence on the connection in the non-Abelian 
case enforces to smear in 1d only \cite{30}. 
Therefore in LQG the Hamitonian constraint acts by attaching loops to a graph 
in the vicinity of a vertex while here the Hamiltonian acts in a sense 
everywhere as the state is non-degenerate. Their action is therefore 
difficult to compare. To examine the closer relation between the two,
we consider the U(1)$^3$ model but {\it exactly} quantised as in LQG
\cite{45}, i.e. in terms of charge networks along 1d graphs rather 
than 3d smearings. Such a charge network smears the connection 
with a form factor 
\be \label{10.1}
(F_{{\rm CNW}})^a_j(x)=\sum_e\; n_e^j \; \int_e dy^a\; \delta(x,y)
\ee
excited on a graph $\gamma$ with edges $e$, charges $n_e^j\in \mathbb{Z}$ 
which 
solve the Gauss constraint $\sum_{b(e)=v} \vec{n}_e=\sum_{f(e)=v} \vec{n}_e$
where $b(e),f(e)$ denote the beginning/final point of an oriented edge 
$e$ and $v$ is a vertex of $\gamma$. 

To obtain a closely related form factor in the present treatment of the 
theory we simply mollify the $\delta$ distribution (see also 
\cite{46} for a similar procedure in linearised gravity)
\be \label{10.2}
(F_\epsilon)^a_j(x)=\sum_e\; n_e^j \; \int_e dy^a\; \delta_\epsilon(x,y)
\ee
where $\epsilon\to \delta_\epsilon$ is a sequence of positive, smooth 
functions of rapid decrease that converge to the $\delta$ distribution 
in the topology of the tempered distributions on $\mathbb{R}^3$, e.g.
a Gaussian. Then it is easy to check that 
\be \label{10.3}
\partial_a (F_\epsilon)^a_j=-\sum_v\;\delta_\epsilon(v)\;
[\sum_{f(e)=v}\; n^j_e-\sum_{b(e)=v}\; n^j_e]=0
\ee
is indeed divergence free and 
\be \label{10.4}
\det(\{F_\epsilon(x)\})=\sum_{e_1,e_2,e_3}\;
\det(\vec{n}_{e_1},\vec{n}_{e_2},\vec{n}_{e_3})\;
\det(F_{\epsilon,e_1},F_{\epsilon,e_2},F_{\epsilon,e_3})(x);\;\;
F^a_{\epsilon,e}(x)=\int_e\; dy^a\; \delta_\epsilon(x,y)
\ee
As $\delta_\epsilon(x,y)$ is sharply peaked at $x=y$ the smooth function
$F^a_{\epsilon,e}(x)$ is essentially supported on $e$ decaying rapidly 
away from it. This shows that to get a non-degenerate $F_\epsilon$ the 
grid $\gamma$ should fill $\mathbb{R}^3$. Such states are outside the 
completion of the CNW states but are included in the present theory. 
However, leaving this issue aside, for $x$ in the vicinity of the 
graph (\ref{10.4}) is dominated by triples of edges $e_1,e_2,e_3$
intersectiong in a vertex $v$. Then for instance for a cubic grid, 
the non-degeneracy is satisfied for suitable charges. Namely if 
we denote the edges of the grid by $e_I(v)$ where $I=1,2,3$ denotes the 
direction and $v-\delta_I$ denotes the vertex shifted from $v$ by one 
lattice unit to the left in driection $Y+I$
then (\ref{10.4}) is dominated by
\be \label{10.5}
\det(
\vec{n}_{e_1(v)}-\vec{n}_1(v-\delta_1),
\vec{n}_{e_2(v)}-\vec{n}_2(v-\delta_2),
\vec{n}_{e_3(v)}-\vec{n}_3(v-\delta_3))\;
\det(F_{\epsilon,e_1(v)},F_{\epsilon,e_2(v)},F_{\epsilon,e_3(v)})(v)
\ee

In the current quantum theory, the action of a single 
Hamiltonian constraint could be considererd as changing the 
``charge network label'' (or form factor) $F^a_j$ to 
\be \label{10.5}
F^a_j+[(X_M K)(F)]^a_j=F^a_j+\Delta F^a_j=F^a_j
+\epsilon_{jkl} \partial_b(M F^b_k F^a_l |\det(F)|^{(\delta-2)/2})        
\ee
where $\delta$ is the density weight used. Expanding (\ref{10.5}) we find
\be \label{10.6}
\Delta F^a_j=\sum_{e_1,e_2}\; \epsilon_{jkl}\ n^k_{e_1} n^l_{e_2}\;
[\partial_b\;(M F^b_{\epsilon,e_1} F^a_{\epsilon,e_2})
|\det(F_\epsilon)|^{(\delta-2)/2})](x)
\ee
where the determinant factor is determined by (\ref{10.4}). Formula
(\ref{10.6}) has the following features:\\
\\
i. For $\delta>0$ (\ref{10.6}) is {\it still concentrated on the  
support of $F$}, i.e. ``the graph is not changed, no loop is 
created''. \\
ii. Rather the charges have changed in a 
very non-linear and even non-algebraic way as there are derivatives 
involved and they change {\it everywhere along the graph}, not only in the 
vicinity of a vertex, although its action in the vicinity 
of the vertex is strongest. As more and more actions are included, the 
derivatives will also eventually move {\it the whole graph}. \\
iii. The lapse function {\it has become 
part of the new form factor}. This is similar to the finite action of the 
spatial diffeomorphisms which map graphs to diffeomorphic images and 
in this sense also become part of the form factor.\\
iv. The whole action of the constraint consists in a change of the 
form factor, that is, it does not map a state $w[F]\Omega$ into a linear
combination of several such states but simply a {\it single state}
$w[F']\Omega$.\\
\\
By contrast, in the current implementation of the Hamiltonian constraint
in LQG we have:\\
\\
i. The constraint acts only in the vicinity of vertices and does 
change the graph there.\\
ii. The charges are only changed in the vicinity of a vertex and 
more and more actions only change the graph ever more closely to 
a vertex.\\     
iii. The lapse function evaluated at vertices 
is a coefficient in an expansion of CNW functions 
and does not get part of the form factor.\\
iv. The image of the constraint, even at a single vertex, is a non-trivial 
expansion of CNW functions with coefficients that depend on inverrse 
volume operator factors.\\
\\
Accordingly, the current paper suggests that the quantisation of the 
Hamiltonian constraint in LQG be changed following the above rules. 
Basically, one should try to exponentiate the Hamiltonian constraint.
This is a non-trivial technical challenge because of the non-Abelian 
gauge group and because the limit $\epsilon\to 0$ of (\ref{10.6}) (returning
to the sharp rather than mollified graph) is singular. However, for 
theories with quadratic dependence on connections, smearing just in 1d
cannot be avoided. This points again to the emphasis on the non-degenerate 
sector of the theory \cite{30}.
\item[6.] {\it Interacting QFT}\\
As a by-product of the present work we have shown:\\
For any spacetime dimension
$1+D$ a field theory on $\mathbb{R}\times \sigma$ ($\sigma$ any D-manifold)
whose classical Hamiltonian $H$ is at most
linear in the momenta $\pi_I$ of the fields $\phi^I$ such that 
both $H,V$ vanish at $\phi^I=0$ where $V$ is the potential of $H$, 
admits a 
quantisation in terms of its 
1-parameter unitary group $t\mapsto U(t)$ in a representation 
of the Weyl algebra generated by $\phi^I,\pi_I$ of Narnhhofer Thirring type
with vacuum $\phi^I\Omega=0$ cyclic for the Weyl operators 
$w[F]=\exp(-i \pi[F]),\;\pi[F]:=\int\; d^Dx\; F^I\; \pi_I$ with suitable 
test functions $F^I$. It is given by 
\be \label{10.7}
U(t)\; w[F]\Omega=e^{-i\alpha_t[F]}\; w[F_t]\;\Omega,\;
\alpha_t[F]=\int_0^t\; ds\; V(F_s),\;F_t=(e^{t\;X_H}\; K)(F)
\ee
where $X_H$ is the Hamiltonian vector field of $H-V$
and where $K^I(\phi,\pi)=\phi^I$ denotes the $I-th$ configuration 
coordinate function. The scalar product is specified by 
$<\Omega,w[F]\Omega>=\delta_{F,0}$. The dependence of $H,V$ on $\phi$ 
can be arbitrarily non-linear, even non-polynomial and can be 
background metric independent. The classical 
Legendre transform of $H$ with respect to $\pi$ is singular (just 
yielding the classical equation of motion for $\phi$), but 
with respect to $\phi$ yields an arbitrarily non-linear Lagrangian 
in terms of $\pi$ and its first time derivative $\dot{\pi}$. The 
resulting N-point functions of this QFT (i.e. vacuum expectation values 
of monomials) of the time translates 
$U(t)\; w[F]\; U(-t)$ are not Gaussians in the test functions  
and thus define an interacting QFT in that sense. Among the drawbacks of 
these QFT's are: The Hilbert space underlying this 
representation of the Weyl algebra is non-separable, the unitary group
is weakly discontinuous, the Hamiltonian is not bounded from below, 
the Lagrangean generically is not Poinvar\'e invariant.    
\end{itemize}

\end{document}